\documentclass[11pt]{article}
\usepackage[margin=1in]{geometry}

\usepackage{indentfirst}
\usepackage{amsmath, amssymb, amsthm, amsfonts}
\usepackage{tikz, comment}
\usepackage{algpseudocode, algorithm}
\usepackage{tcolorbox, qtree}
\usepackage{enumerate}
\usepackage{fancyvrb}
\usepackage{url}
\usepackage{caption}
\usepackage{xfrac}

\DeclareMathOperator*{\argmax}{arg\,max}

\newtheorem{theorem}{Theorem}

\newtheorem{lemma}{Lemma}

\newcount\Comments  
\newcommand{\kibitz}[2]{\ifnum\Comments=1{\color{#1}{#2}}\fi}
\newcommand{\ns}[1]{\kibitz{blue}{[Nisarg: #1]}}
\newcommand{\bdf}[1]{\kibitz{red}{[Brandon: #1]}}
\newcommand{\km}[1]{\kibitz{orange}{[Kamesh: #1]}}

\newtheorem{example}{Example}
\newcommand{\E}{\mathbf{E}}
\newcommand{\set}[1]{\left\{#1\right\}}
\newcommand{\calC}{\mathcal{C}}
\newcommand{\bbR}{\mathbb{R}}
\newcommand{\M}{\mathcal{M}}
\newcommand{\calF}{\mathcal{F}}

\newcommand{\bc}{{\mathbf c}} 
\newcommand{\bcstar}{{\mathbf c^*}} 
\newcommand{\bchat}{\widehat{\mathbf c}} 
\newcommand{\bcp}{{\mathbf c'}}

\newcommand{\bd}{{\mathbf d}}
\renewcommand{\vec}{\mathbf}

\newcommand{\gain}{\mathrm{gain}}

\newtheorem{definition}{Definition}

\begin{document}

\title{Fair Allocation of Indivisible Public Goods}

\author{Brandon Fain\thanks{Department of Computer Science, Duke University, Durham, NC 27708. {\tt btfain@cs.duke.edu}.}  
\and Kamesh Munagala\thanks{Department of Computer Science, Duke University, Durham NC 27708. {\tt kamesh@cs.duke.edu}.} 
\and Nisarg Shah\thanks{Department of Computer Science, University of Toronto, Toronto ON, M5S3G8, Canada. {\tt nisarg@cs.toronto.edu}}}

\date{}


\maketitle

\begin{abstract}
We consider the problem of fairly allocating indivisible public goods. We model the public goods as elements with feasibility constraints on what subsets of elements can be chosen, and assume that agents have additive utilities across elements. Our model generalizes existing frameworks such as fair public decision making and participatory budgeting. We study a groupwise fairness notion called the {\em core}, which generalizes well-studied notions of proportionality and Pareto efficiency, and requires that each subset of agents must receive an outcome that is fair relative to its size. 


In contrast to the case of divisible public goods (where fractional allocations are permitted), the core is not guaranteed to exist when allocating indivisible public goods. Our primary contributions are the notion of an {\em additive} approximation to the core (with a tiny multiplicative loss), and polynomial time algorithms that achieve a small additive approximation, where the additive factor is relative to the largest utility of an agent for an element. If the feasibility constraints define a matroid, we show an additive approximation of $2$. A similar approach yields a constant additive bound when the feasibility constraints define a matching. More generally, if the feasibility constraints define an arbitrary packing polytope with mild restrictions, we show an additive guarantee that is logarithmic in the {\em width} of the polytope. Our algorithms are based on variants of the convex program for maximizing the Nash social welfare, but differ significantly from previous work in how it is used. 
Our guarantees are meaningful even when there are fewer elements than the number of agents. As far as we are aware, our work is the first to approximate the core in indivisible settings. 
\end{abstract}


\section{Introduction}
\label{sec:intro}

In fair resource allocation, most work considers \textit{private goods}; each good must be assigned to a particular agent (and no other). However, not all goods are private. \textit{Public goods} are those which can be enjoyed by multiple agents simultaneously, like a public road. Allocation of public goods generalizes the problem of allocation of private goods, and, as we will see, can provide new difficulties from both a normative and an algorithmic perspective.

Consider an example to highlight what a public resource allocation problem might look like, and why fairness might be a concern. Suppose that the next time you vote, you see that there are four referendums for your consideration on the ballot, all of which concern the allocation of various public goods in your city: A = a new school, B = enlarging the public library, C = renovating the community college, and D = improving a museum. In 2016, residents of Durham, North Carolina faced precisely these options~\cite{DurhamElection}. Suppose the government has resources to fund only two of the four projects, and the (hypothetical) results were as follows: a little more than half of the population voted for $(A, B)$, a little less than half voted for $(C,D)$, and every other combination received a small number of votes. Which projects should be funded?

If we na\"{\i}vely tally the votes, we would fund A and B, and ignore the preferences of a very large minority. In contrast, funding A and C seems like a reasonable compromise. Of course, it is impossible to satisfy \textit{all} voters, but given a wide enough range of possible outcomes, perhaps we can find one that fairly reflects the preferences of large subsets of the population. This idea is not captured by fairness axioms like proportionality or their approximations~\cite{FPDM}, which view fairness from the perspectives of {\em individual} agents. Indeed, in the aforementioned example, {\em every} allocation gives zero utility to {\em some} agent, and would be deemed equally good according to such fairness criteria.

\subsection{Public Goods Model}
We consider a fairly broad model for public goods allocation that generalizes much of previous work~\cite{LMMS04,FPDM,Fain2016,Brill,FairKnapsack,envyFreeUpTo1,envyFreeUpToAny}.  There is a  set of voters (or agents) $N = [n]$. Public goods are modeled as elements of a ground set $W$. We denote $m = |W|$. An {\em outcome} $\bc$ is a subset of $W$. 
Let $\calF \subseteq 2^W$ denote the set of feasible outcomes. 

The utility of agent $i$ for element $j \in W$ is denoted $u_{ij} \in \bbR_{\ge 0}$. We assume that agents have additive utilities, i.e., the utility of agent $i$ under outcome $\bc \in \calF$ is $u_i(\bc) = \sum_{j \in \bc} u_{ij}$. Since we are interested in scale-invariant guarantees, we assume without loss of generality that $\max_{j \in W} u_{ij} = 1$ for each agent $i$, so that $u_{ij} \in [0,1]$ for all $i,j$. Crucially, this does not restrict the utility of an agent for an outcome to be $1$: $u_i(\bc)$ can be as large as $m$. Specifically, let $V_i = \max_{\bc \in \calF} u_i(\bc)$, and $V_{\max} = \max_{i\in N} V_i$. Our results differ by the feasibility constraints imposed on the outcome. We consider three types of constraints, special cases of which have been studied previously in literature.

\paragraph{Matroid Constraints.} In this setting, we are given a matroid $\M$ over the ground set $W$, and the feasibility constraint is that the chosen elements must form a basis of $\M$ (see~\cite{Kung} for a formal introduction to matroids). 


This generalizes the {\em public decision making} setting introduced by~\cite{FPDM}. In this setting, there is a set of issues $T$, and each issue $t \in T$ has an associated set of alternatives $A^t = \{a_1^t, \hdots, a_{k_t}^t\}$, exactly one of which must be chosen. Agent $i$ has utility $u_i^t(a_j^t)$ if alternative $a_j^t$ is chosen for issue $t$, and utilities are additive across issues.  An outcome $\bc$ chooses one alternative for every issue. It is easy to see that if the ground set is $\cup_t A^t$, the feasibility constraints correspond to a partition matroid. We note that public decision making in turn generalizes the classical setting of {\em private goods allocation}~\cite{LMMS04,envyFreeUpTo1,envyFreeUpToAny} in which private goods must be divided among agents with additive utilities, with each good allocated to exactly one agent. 

Matroid constraints also capture multi-winner elections in the voting literature (see, e.g.~\cite{Brill}), in which voters have additive utilities over candidates, and a committee of at most $k$ candidates must be chosen. This is captured by a uniform matroid over the set of candidates.



\paragraph{Matching Constraints.}  In this setting, the elements are edges of an undirected graph $G(V,E)$, and the feasibility constraint is that the subset of edges chosen must form a {\em matching}. Matchings constraints in a bipartite graph can be seen as the intersection of two matroid constraints. Matching constraints are a special case of the more general packing constraints we consider below. 

\paragraph{Packing Constraints.} In this setting, we impose a set of packing constraints $A \vec{x} \le \vec{b}$, where $x_j \in \{0,1\}$ is the indicator denoting whether element $j$ is chosen in the outcome. Suppose $A$ is a $K \times m$ matrix, so that there are $K$ packing constraints. By scaling, we can assume $a_{kj} \in [0,1]$ for all $k,j$. Note that even for one agent, packing constraints encode independent set. Thus, to make the problem tractable, we assume $\vec{b}$ is sufficiently large, in particular, $b_k = \omega\left( \log K  \right)$ for all $k \in \{1,2,\ldots,K\}$. This is in contrast to matroid and matching constraints, for which single-agent problems are polynomial time solvable. A classic measure of how easy it is to satisfy the packing constraints is the {\em width} $\rho$~\cite{PST}:
\begin{equation}
\label{eq:width} \rho = \max_{k \in [K]} \frac{\sum_{j \in [m]} a_{kj}}{b_k}.
\end{equation}

Packing constraints capture the general {\sc Knapsack} setting, in which there is a set of $m$ items, each item $j$ has an associated size $s_j$, and a set of items of total size at most $B$ must be selected. This setting is motivated by participatory budgeting applications~\cite{PBP,knapsack1,knapsack2,Fain2016,GargKGMM17,FairKnapsack,BNPS17}, in which the items are public projects, and the sizes represents the costs of the projects. {\sc Knapsack} uses a single packing constraint. Multiple packing constraints can arise if the projects consume several types of resources, and there is a budget constraint for each resource type. For example, consider a statewide participatory budgeting scenario where each county has a budget than can only be spent on projects affecting that county, the state has some budget that can be spent in any county, and projects might affect multiple counties. In such settings, it is natural to assume a small width, i.e., that the budget for each resource is such that a large fraction (but not all) of the projects can be funded. We note that the aforementioned multi-winner election problem is a special case of the {\sc Knapsack} problem with unit sizes. 

\subsection{Prior Work: Fairness Properties}
We define fairness desiderata for the public goods setting by generalizing appropriate desiderata from the private goods setting such as Pareto optimality, which is a weak notion of efficiency, and proportionality, which is a per-agent fair share guarantee.\footnote{Those familiar with the literature on fair division of private goods will note the conspicuous absence of the \textit{envy freeness} property: that no agent should (strongly) prefer the allocation of another agent.  Because we are considering public goods, envy freeness is only vacuously defined: the outcome in our setting is common to all agents.}

\begin{definition}
An outcome $\bc$ satisfies \textbf{Pareto optimality}  if there is no other outcome $\bcp$ such that $u_i(\bcp) \geq u_i(\bc)$ for all agents $i \in N$, and at least one inequality is strict.
\end{definition}

Recall that $V_i$ is the maximum possible utility agent $i$ can derive from a feasible outcome. 

\begin{definition}
The \textbf{proportional share} of an agent $i \in N$ is $Prop_i := \frac{V_i}{n}$.  For $\beta \in (0,1]$, we say that an outcome $\bc$ satisfies $\beta$-proportionality if $u_i(\bc) \geq \beta \cdot Prop_i$ for all agents $i \in N$.  If $\beta=1$, we simply say that $\bc$ satisfies proportionality.
\end{definition}

The difficulty in our setting stems from requiring integral outcomes, and not allowing randomization. In the absence of randomization, it is reasonably straightforward to show that we cannot guarantee $\beta$-proportionality for any $\beta \in (0,1]$. Consider a problem instance with two agents and two feasible outcomes, where each outcome gives a positive utility to a unique agent. In any feasible outcome, one agent has zero utility, which violates $\beta$-proportionality for every $\beta > 0$. 

To address this issue, \cite{FPDM} introduced the novel relaxation of proportionality up to one issue in their public decision making framework, inspired by a similar relaxation called envy-freeness up to one good in the private goods setting~\cite{LMMS04,envyFreeUpTo1}. 
%
They say that an outcome $\bc$ of a public decision making problem satisfies \emph{proportionality up to one issue} if for all agents $i \in N$, there exists an outcome $\bcp$ that differs from $\bc$ only on a single issue and $u_i(\bcp) \geq Prop_i$. Proportionality up to one issue is a reasonable fairness guarantees only when the number of issues is larger than the number of agents; otherwise, it is vacuous and is satisfied by all outcomes. 
Thus, it is perfectly reasonable for some applications (e.g., three friends choosing a movie list to watch together over the course of a year), but not for others (e.g., when thousands of residents choose a handful of public projects to finance). In fact, it may produce an outcome that may be construed as unfair if it does not reflect the wishes of large groups of voters. Thus, in this work, we address the following question posed by~\cite{FPDM}: 

\begin{quote}{\em Is there a stronger fairness notion than proportionality in the public decision making framework...? Although such a notion would not be satisfiable by deterministic mechanisms, it may be satisfied by randomized mechanisms, or it could have novel relaxations that may be of independent interest.} 
\end{quote} 

\subsection{Summary of Contributions}
Our primary contributions are twofold.
\begin{itemize}
	\item We define a fairness notion for public goods allocation that is stronger than proportionality, ensures fair representation of groups of agents, and in particular, provides a meaningful fairness guarantee even when there are fewer goods than agents.
	\item We provide polynomial time algorithms for computing integer allocations that approximately satisfy this fairness guarantee for a variety of settings generalizing the public decision making framework and participatory budgeting.
\end{itemize}

\subsection{Core and Approximate Core Outcomes} 
Below, we define the notion of \emph{core outcomes}, which has been extensively studied (in similar forms) as a notion of stability in economics~\cite{lindahlCore,scarfCore,coreConjectureCounter} and computer science~\cite{ROBUS,Fain2016} in the context of {\em randomized} or {\em fractional} allocations. Our main contribution is to study it in the context of {\em integer} allocations. 

\begin{definition}
\label{def:core}
Given an outcome $\bc$, we say that a set of agents $S \subseteq N$ form a blocking coalition if there exists an outcome $\bcp$ such that $(|S|/n) \cdot u_i(\bcp) \geq u_i(\bc)$ for all $i \in S$ and at least one inequality is strict. 
We say that an outcome $\bc$ is a {\bf core outcome} if it admits no blocking coalitions.
\end{definition}

Note that non-existence of blocking coalitions of size $1$ is equivalent to proportionality, and non-existence of blocking coalitions of size $n$ is equivalent to Pareto optimality. Hence, a core outcome is both proportional and Pareto optimal. However, the core satisfies a stronger property of being, in a sense, Pareto optimal for coalitions of {\em any size}, provided we scale utilities based on the size of the coalition. 
Another way of thinking about the core is to view it as a fairness property that enforces a proportionality-like guarantee for coalitions: \textit{e.g.}, if half of all agents have identical preferences, they should be able to get at least half of their maximum possible utility. It is important to note that the core provides a guarantee for every possible coalition. Hence, in satisfying the guarantee for a coalition $S$, a solution cannot simply make a single member $i \in S$ happy and ignore the rest as this would likely violate the guarantee for the coalition $S \setminus \{i\}$.



\paragraph{Approximate Core.} Since a proportional outcome is not guaranteed to exist (even allowing for multiplicative approximations), the same is true for the core. However, an additive approximation to the core still provides a meaningful guarantee, even when there are fewer elements than agents because it provides a non-trivial guarantee to large coalitions of like-minded agents. 
\begin{definition}
\label{def:approx}
For $\delta,\alpha \ge 0$, an outcome $\bc$ is a {\bf $\boldsymbol{(\delta, \alpha)}$-core outcome} if there exists no set of agents $S \subseteq N$ and outcome $\bcp$ such that 
$$
\frac{|S|}{n}\cdot u_i(\bcp) \ge (1+\delta) \cdot u_i(\bc) + \alpha
$$
for all $i \in S$, and at least one inequality is strict.
\end{definition}

A $(0,0)$-core outcome is simply a core outcome. A $(\delta,0)$-core outcome satisfies $\delta$-proportionality. Similarly, a $(0,1)$-core outcome $\bc$ satisfies the following relaxation of proportionality that is slightly weaker than proportionality up to one issue: for every agent $i \in N$, $u_i(\bc)+1 \ge Prop_i$. We note that this definition, and by extension, our algorithms satisfy scale invariance, i.e., they are invariant to scaling the utilities of any individual agent. Because we normalize utilities of the agents, the true additive guarantee is $\alpha$ times the maximum utility an agent can derive from a single element. Since an outcome can have many elements, an approximation with small $\alpha$ remains meaningful.

The advantage of an approximate core outcome is that it fairly reflects the will of a like-minded subpopulation relative to its size. An outcome satisfying approximate proportionality only looks at what {\em individual} agents prefer, and may or may not respect the collective preferences of sub-populations. We present such an instance in Example~\ref{eg:prop1} (Section~\ref{sec:fnw}), in effect showing that an approximate core outcome is arguably more fair.

In our results, we will assume $\delta < 1$ to be a small constant, and focus on making $\alpha$ as small as possible. In particular, we desire guarantees on $\alpha$ that exhibit sub-linear or no dependence on $n$, $m$, or any other parameters. Deriving such bounds is the main technical focus of our work. 

\subsection{Our Results} 
We present algorithms to find approximate core outcomes under matroid, matching, and general packing constraints. Our first result (Section~\ref{sec:mnw}) is the following:

\begin{theorem}
\label{thm:matroid}
If feasible outcomes are constrained to be bases of a matroid, then a $(0,2)$-core outcome is guaranteed to exist, and for any $\epsilon > 0$, a $(0,2+\epsilon)$-core outcome can be computed in time polynomial in $n, m,$ and $\sfrac{1}{\epsilon}$.
\end{theorem}

In particular, for the public decision making framework, the private goods setting, and multi-winner elections (a.k.a. {\sc Knapsack} with unit sizes), there is an outcome whose guarantee for {\em every coalition} is close to the guarantee that~Conitzer et al. provide to individual agents~\cite{FPDM}.  

In Section~\ref{sec:matching}, we consider matching constraints. 
Our result now involves a tradeoff between the multiplicative and additive guarantees.

\begin{theorem}
\label{thm:matching}
If feasible outcomes are constrained to be matchings in an undirected graph, then for constant $\delta \in (0,1]$, a $\left(\delta,8+\sfrac{6}{\delta}\right)$-core outcome can be computed in time polynomial in $n$ and $m$.
\end{theorem}


\medskip
Our results in Section~\ref{sec:main} are for general packing constraints. Here, our guarantee depends on the width $\rho$ from Equation~(\ref{eq:width}), which captures the difficulty of satisfying the constraints. In particular, the guarantee improves if the constraints are easier to satisfy. This is the most technical result of the paper, and involves different techniques than those used in proving Theorems~\ref{thm:matroid} and~\ref{thm:matching}; we present an outline of the techniques in Section~\ref{sec:idea}. 

\begin{theorem}
\label{thm:packing}
For constant $\delta \in (0,1)$, given $K$ packing constraints $A \vec{x} \le \vec{b}$ with width $\rho$ and $b_k = \omega \left(\frac{\log K}{\delta^2}\right)$ for all $k \in [K]$,  there exists a polynomial time computable $(\delta,\alpha)$-core solution, where 
$$
\alpha = O\left( \frac{1}{\delta^4}  \cdot \log\left(\frac{\min(V_{\max},n,\rho) \cdot \log^* V_{\max}}{\delta} \right)\right).
$$
\end{theorem}

Here, $\log^*$ is the iterated logarithm, which is the number of times the logarithm function must be iteratively applied before the result becomes less than or equal to $1$.  Recall that $V_{\max}$ is the maximum utility an agent can have for an outcome (thus $V_{\max} \le m$); our additive error bound is a vanishing fraction of this quantity. Our bound is also small if the number of agents $n$ is small. Finally, the guarantee improves for small $\rho$, i.e., as the packing constraints become easier to satisfy. For instance, in participatory budgeting, if the total cost of all projects is only a constant times more than the budget, then our additive guarantee is close to a constant.  

Note that $V_{\max}$ (which is bounded by $m$), $n$, and $\rho$ are all unrelated quantities --- either could be large with the other two being small. In fact, in Section~\ref{sec:main}, we state the bound more generally in terms of what we call the {\em maximally proportionally fair value} $R$, which informally captures the (existential) difficulty of finding a proportionally fair allocation. The quantity $\min(V_{\max},n,\rho)$ stems from three different bounds on the value of $R$.

In Example~\ref{eg:IS} (Appendix~\ref{sec:Examples}), we show that the lower bound on $\vec{b}$ in the above theorem is necessary: if $\vec{b} = O(1)$, then no non-trivial approximation to the core can be guaranteed, even when $\rho$ is a constant. 

\medskip
Finally, in Appendix~\ref{sec:prop}, we consider a different (and more classical) version of the core for general packing constraints, in which a deviating coalition gets a proportional share of resources rather than a proportional share of utility. 
We show that our techniques provide a similar approximation to this version of the core, although we do not provide an efficient algorithm in this model. 

\subsection{Related Work}
\paragraph{Core for Public Goods.}  The notion of core is borrowed from cooperative game theory and was first phrased in game theoretic terms by~\cite{scarfCore}. It has been extensively studied in public goods settings~\cite{lindahlCore,coreConjectureCounter,Fain2016}. Most literature so far has considered the core with {\em fractional allocations}. Our definition of core (Definition~\ref{def:core}) assumes the utility of a deviating coalition is scaled by the size of the coalition. For fractional allocations, one such core allocation coincides with the well-known notion of {\em proportional fairness}, the extension of the Nash bargaining solution~\cite{nashBargaining}. This solution maximizes the product of the utilities of the agents, and we present the folklore proof in Section~\ref{sec:fractional}.  Our main focus is on finding {\em integer} allocations that approximate the core, and to the best of our knowledge, this has not been studied previously.  

A simpler property than the core is proportionality, which like the core, is impossible to satisfy to any multiplicative approximation using integral allocations. To address this problem, \cite{FPDM} defined proportionality up to one issue in the public decision making framework, inspired by related notions for private goods. 
This guarantee is satisfied by the integral outcome maximizing the {\em Nash welfare} objective, which is the {\em geometric} mean of the utilities to the agents. For public goods, this objective is not only {\sc NP-Hard} to approximate to any multiplicative factor, but approximations to the objective also do not retain the individual fairness guarantees.

We extend the notion of additive approximate proportionality to additive approximate core outcomes, which provides meaningful guarantees even when there are fewer goods than agents. Unlike proportionality, we show in Section~\ref{sec:fnw} that the approach of computing the optimal integral solution to the Nash welfare objective fails to provide a reasonable approximation to the core.  Therefore, for our results about matroid constraints (Theorem~\ref{thm:matroid}) and matching constraints (Theorem~\ref{thm:matching}), we slightly modify the integer Nash welfare objective and add a suitable constant term to the utility of each agent. We show that maximizing this smooth objective function achieves a good approximation to the core.  However, maximizing this objective is still {\sc NP-hard}~\cite{FairKnapsack}, so we devise local search procedures that run in polynomial time and still give good approximations of the core. In effect, we make a novel connection between appropriate {\em local optima} of smooth Nash Welfare objectives and the core.

\paragraph{Fairness on Endowments.} Classically, the core is defined in terms of agent endowments, not scaled utilities. In more detail, in Definition~\ref{def:core}, we assumed that when a subset $S$ of agents deviates, they can choose any feasible outcome; however, their utility is reduced by a factor that depends on $|S|$. A different notion of core is based on endowments~\cite{scarfCore,lindahlCore} and has been considered in the context of participatory budgeting~\cite{Fain2016} and in {\em proportional representation} of voters in multi-winner elections with approval voting. In this notion, a deviating coalition gets a proportional share of resources rather than a proportional share of utility. For example, if the elements have different sizes, and we need to select a subset of them with total size at most $B$, then a deviating coalition $S$ would get to choose an outcome with total size at most $B|S|/n$ instead of $B$, but would not have its utility scaled down. This notion builds on the seminal work of Foley on the Lindahl equilibrium~\cite{lindahlCore}, from which it can be shown that such a core outcome always exists when fractional allocations are allowed. However, it is not known how to compute such a core outcome efficiently, and further, it is difficult to define such a notion of endowments in settings such as matroid or matching constraints. In the context of integer allocations with packing constraints, we extend our techniques to provide approximations to the notion of core with endowments in Appendix~\ref{sec:prop}, though this is not the main focus of our paper.

The notion of core with endowments logically implies a number of fairness notions considered in multi-winner election literature, such as justified representation, extended justified representation~\cite{Brill}, and proportional justified representation~\cite{Sanchez}. Approval-based multi-winner elections are a special case of packing constraints, in which voters (agents) have binary utilities over a pool of candidates (elements), and we must select a set of at most $B$ candidates. The idea behind proportional representation is to define a notion of large cohesive groups of agents with similar preferences, and ensure that such coalitions are proportionally represented. The core on endowments represents a more general condition that holds for all coalitions of agents, not just those that are large and cohesive. Nevertheless, our local search algorithms for Theorems~\ref{thm:matroid} and~\ref{thm:matching} are similar to local search algorithms for proportional approval voting (PAV)~\cite{Thiele,PJR2018} that achieve proportional representation. It would be interesting to explore the connection between these various notions in greater depth.

\paragraph{Private Goods and Envy-freeness.} Private goods are a special case of public goods with matroid constraints. Fair allocation of private goods is a widely studied topic~\cite{varian,drf,Parkes:2012:BDR:2229012.2229075,kellyTCP}. A common fairness criterion for private goods is envy-freeness: that no agent should (strongly) prefer the allocation of another agent. For fractional allocations, the classic context for envy-free allocation is cake cutting~\cite{cakeCutting, cakeCuttingProtocol}. For integral allocations, envy-free allocations or multiplicative approximations thereof may not exist in general. Recent work has introduced envy-freeness up to one good~\cite{LMMS04,budish2,envyFreeUpTo1,envyFreeUpToAny}, an additive approximation of envy-freeness. The notion of envy does not extend as is to public goods, and the core can be thought of as enforcing envy-freeness across demographics. 
We note that in addition to resource allocation, group based fairness is also appearing as a desideratum in machine learning. Specifically, related notions may provide a tool against gerrymandered classifiers that appear fair on small samples, but not on structured subsets~\cite{MLFairness}.

\paragraph{Strategyproofness.} In this work, we will not consider game-theoretic incentives for manipulation  for two reasons. First, even for the restricted case of private goods allocation, preventing manipulation leads to severely restricted mechanisms. For instance, \cite{Schum97} shows that the only strategyproof and Pareto efficient mechanisms are dictatorial, and thus highly unfair, even when there are only two agents with additive utilities over divisible goods. Second, our work is motivated by public goods settings with a large number of agents, such as participatory budgeting, wherein individual agents often have limited influence over the final outcome. It would be interesting to establish this formally, using notions like {\em strategyproofness in the large}~\cite{AB17}.

\section{Prelude: Nash Social Welfare}
Our approach to computing approximate core solutions revolves around the Nash social welfare, which is the product (or equivalently, the sum of logarithms) of agent utilities. This objective is commonly considered to be a natural tradeoff between the fairness-blind utilitarian social welfare objective (maximizing the sum of agent utilities) and the efficiency-blind egalitarian social welfare objective (maximizing the minimum agent utility). This function also has the advantage of being \textit{scale invariant} with respect to the utility function of each agent, and in general, preferring more equal distributions of utility. 

\subsection{Integer Nash Welfare and Smooth Variants}
\label{sec:fnw}
The integer {\em Max Nash Welfare} (MNW) solution~\cite{envyFreeUpTo1,FPDM} is an outcome $\bc$ that maximizes $\sum_{i \in N} \ln u_i(\bc) $.  More technically, if every integer allocation gives zero utility to at least one agent, the MNW solution first chooses a largest set $S$ of agents that can be given non-zero utility simultaneously, and maximizes the product of utilities to agents in $S$.


\cite{FPDM} argue that this allocation is reasonable by showing that it satisfies proportionality up to one issue for public decision making. A natural question is whether it also provides an approximation of the core. {\em We answer this question in the negative.} The example below shows that even for public decision making (a special case of matroid constraints), the integer MNW solution may fail to return a $(\delta,\alpha)$-core outcome, for any $\delta = o(m)$ and $\alpha = o(m)$.

\begin{example}
\label{eg:prop1}
\label{eg:mnw}

\normalfont{	Consider an instance of public decision making~\cite{FPDM} with $m$ issues and two alternatives per issue. Specifically, each issue $t$ has two alternatives $\{a_1^t, a_2^t\}$, and exactly one of them needs to be chosen. There are two sets of agents $X = \set{1,\ldots,m}$ and $Y = \set{m+1,\ldots,2m}$. Every agent $i \in X$ has $u_i^i(a_1^i) = 1$, and utility 0 for all other alternatives. Every agent $i \in Y$ has $u_i^t(a_2^t) = 1$ and $u_i^t(a_1^t) = 1/m$ for all issues $i \in \{1,2,\ldots,m\}$.  Visually, this is represented as follows.

\begin{center}
	\begin{tabular}{ c | c c | c c | c | c c}
		& $a_1^1$ & $a_2^1$ & $a_1^2$ & $a_2^2$ & $\hdots$ & $a_1^m$ & $a_2^m$ \\
		\hline
		$u_{1 \in X}$ & 1 & 0 & 0 & 0 & $\hdots$ & 0 & 0 \\
		$u_{2 \in X}$ & 0 & 0 & 1 & 0 & $\hdots$ & 0 & 0 \\
		\vdots & \vdots & \vdots & \vdots & \vdots &  & \vdots & \vdots \\
		$u_{m \in X}$ & 0 & 0 & 0 & 0 & $\hdots$ & 1 & 0 \\
		$u_{i \in Y}$ & $1/m$ & 1 & $1/m$ & 1 & $\hdots$ & $1/m$ & 1
	\end{tabular}
\end{center}

The integer MNW outcome is $\mathbf{c} = (a_1^1, a_1^2, \hdots, a_1^{m})$ because any other outcome gives zero utility to at least one agent. However, coalition $Y$ can deviate, choose outcome $\mathbf{c'} = (a_2^1, a_2^2, \hdots, a_2^{m})$, and achieve utility $m$ for each agent in $Y$. For $\mathbf{c}$ to be a $(\delta,\alpha)$-core outcome, we need
\begin{align*}
\exists i \in Y : (1+\delta) \cdot u_i(\mathbf{c}) + \alpha \ge \frac{|Y|}{|Y|+|X|} \cdot u_i(\mathbf{c'}) \quad \Rightarrow \quad 1+\delta + \alpha \ge \frac{m}{2}.
\end{align*}
Hence, $\mathbf{c}$ is not a $(\delta,\alpha)$-core outcome for any $\delta = o(m)$ and $\alpha = o(m)$. In contrast, it is not hard to see that $\mathbf{c'}$ is a $(0,1)$-core outcome because each agent in $X$ gets utility at most one in any outcome. 

Further, note that outcome $\mathbf{c}$ gives every agent utility $1$. Since $Prop_i \le 1$ for each agent $i$, $\mathbf{c}$ satisfies proportionality, and yet fails to provide a reasonable approximation to the core. One may argue that $\mathbf{c'}$, which is a $(0,1)$-core outcome, is indeed fairer because it respects the utility-maximizing choice of half of the population; the other half of the population cannot agree on what they want, so respecting their top choice is arguably a less fair outcome. Hence, the example also shows that outcomes satisfying proportionality (or proportionality up to one issue) can be very different from and less fair than approximate core outcomes.  }

\end{example}

\paragraph{Smooth Nash Welfare.}
One issue with the Nash welfare objective is that it is sensitive to agents receiving zero utility. We therefore consider the following smooth Nash welfare objective:
\begin{equation}
\label{eq:FK}F(\bc) :=  \sum_{i \in N}  \ln \left(\ell +u_i(\bc) \right)  
\end{equation}
where $\ell \ge 0$ is a parameter. Note that $\ell = 0$ coincides with the Nash welfare objective. The case of $\ell = 1$ was considered by~\cite{FairKnapsack}, who showed it is {\sc NP-Hard} to optimize. Recall that we normalized agent utilities so that each agent has a maximum utility of $1$ for any element, so when we add $\ell$ to the utility of agent $i$, it is equivalent to adding $\ell \max_j u_{ij}$ to the utility of agent $i$ when utilities are not normalized. 

We show that local search procedures for the smooth Nash welfare objective, for appropriate choices of $\ell$, yield a $(0,2)$-core outcome for matroid constraints (Section~\ref{sec:mnw}) and a $\left(\delta, O\left(\frac{1}{\delta}\right)\right)$-core outcome for matching constraints (Section~\ref{sec:matching}). In contrast, in Example~\ref{eg:knapsack} (Appendix~\ref{sec:Examples}) we show that optimizing any fixed smooth Nash welfare objective cannot guarantee a good approximation to the core, even with {\em a single} packing constraint, motivating the need for a different algorithm.

\renewcommand{\P}{\mathcal{P}}
\subsection{Fractional Max Nash Welfare Solution}
\label{sec:fractional}
For general packing constraints, we use a fractional relaxation of the Nash welfare program. A {\em fractional outcome} consists of a vector $\mathbf{w}$ such that $w_j \in [0,1]$ measures the fraction of element $j$ chosen. The utility of agent $i$ under this outcome is $u_i(\mathbf{w}) = \sum_{j=1}^m w_j u_{ij}$. 
The fractional {\em Max Nash Welfare} (MNW) solution is a fractional allocation that maximizes the Nash welfare objective (without any smoothing). Define the packing polytope as:  
$$\P = \left\{ \mathbf{w} \in [0,1]^m \ | \ \textstyle\sum_{j=1}^{m} a_{kj}  w_j  \leq   b_k,  \forall k \in [K] \right\}$$
Then the fractional MNW solution is $\argmax_{\bc \in \P} \sum_i \ln u_i(\bc)$. 

It is easy to show that the fractional MNW allocation lies in the core.  Let $\bc$ denote the optimal fractional allocation to the MNW program. By first order optimality, for any other allocation $\bd$, 
\begin{equation}
\nabla  \ln \vec{u(\bc)} \cdot \left(\vec{u(\bd)} - \vec{u(\bc)} \right) \le 0 \quad \Rightarrow \quad \sum_{i \in N} \frac{1}{u_i(\bc)} \left(u_i(\bd) - u_i(\bc) \right) \leq 0 \quad \Rightarrow \quad   \sum_{i \in N} \frac{u_i(\bd)}{u_i(\bc)} \leq n.
\label{eqn:prop-fairness}
\end{equation}

Suppose for contradiction that $\bc$ is not a core outcome.  Then there exists a set of agents $S \subseteq N$ and an outcome $\bd$ such that $u_i(\bd) \ge (n/|S|) \cdot u_i(\bc)$, and at least one inequality is tight. This implies $\sum_{i \in S} u_i(\bd)/u_i(\bc) > n$. However, this contradicts Equation~\eqref{eqn:prop-fairness}. Thus $\bc$, the optimal fractional solution to the MNW program, is a core solution.

\medskip
For the allocation of public goods, it can be shown that the fractional MNW outcome can be irrational despite rational inputs~\cite{AACK+17}, preventing an exact algorithm. For our approximation results, a fractional solution that approximately preserves the utility to each agent would suffice, and we prove the following theorem in Appendix~\ref{app:fractional}. 

\begin{theorem} 
\label{thm:fractional}
For any $\epsilon, \delta > 0$, we can compute  a fractional $(\delta,\epsilon)$-core outcome in time polynomial in the input size and $\log \frac{1}{\epsilon \delta}$.
\end{theorem}
\section{Matroid Constraints}
\label{sec:mnw}
We now consider public goods allocation with matroid constraints. In particular, we show that when the feasibility constraints encode independent sets of a matroid $\M$, maximizing the smooth Nash welfare objective in Equation~\eqref{eq:FK} with $\ell = 1$ yields a $(0,2)$-core outcome. However, optimizing this objective is known to be {\sc NP-hard}~\cite{FairKnapsack}. We also show that given $\epsilon > 0$, a local search procedure for this objective function (given below) yields a $(0,2+\epsilon)$-core outcome in polynomial time, which proves Theorem~\ref{thm:matroid}.

\subsection{Algorithm}
Fix $\epsilon > 0$. Let $\gamma = \frac{\epsilon}{4m}$, where $m = |W|$ is the number of elements. Recall that there are $n$ agents. 
\begin{enumerate}
\item Start with an arbitrary basis  $\bc$ of $\M$.
\item Compute $F(\bc) = \sum_{i \in N} \ln (1+ u_i(\bc))$. 
\item Let a {\em swap} be a pair $(j,j')$ such that $j \in \bc$, $j' \notin \bc$, and $\bcp = \bc \setminus \{j\} \cup \{j'\}$ is also a basis of $\M$. 
\item Find a swap such that $F(\bcp)-F(\bc) \ge \frac{n \gamma}{m}$.
\begin{itemize}
\item If such a swap exists, then perform the swap, i.e., update $\bc \gets \bcp$, and go to Step (2).
\item If no such swap exists, then output $\bc$ as the final outcome.
\end{itemize}
\end{enumerate}

\subsection{Analysis}
First, we show that the local search algorithm runs in time polynomial in $n$, $m$, and $\sfrac{1}{\epsilon}$. Note that $F(\bc) = O(n \ln m)$ because in our normalization, each agent can have utility at most $m$. Thus, the number of iterations is $O\left(m^2 \log m/ \epsilon\right)$. Finally, each iteration can be implemented in $O(n \cdot m^2)$ time by iterating over all pairs and computing the change in the smooth Nash welfare objective.


Next, let $\bcstar$ denote the outcome maximizing the smooth Nash welfare objective with $\ell = 1$, and $\bchat$ denote the outcome returned by the local search algorithm. We show that $\bcstar$ is a $(0,2)$-core outcome, while $\bchat$ is a $(0,2+\epsilon)$-core outcome. 

For outcome $\bc$, define $ F_i(\bc) = \ln (1+ u_i(\bc))$. Fix an arbitrary outcome $\bc$. For an agent $i$ with $u_i(\bc) > 0$, we have that for every element $j \in \bc$:
$$
F_i(\bc) - F_i(\bc \setminus \{j\}) \le \frac{u_{ij}}{ u_i(\bc) + 1 - u_{ij}} \le \frac{u_{ij}}{ u_i(\bc)}.
$$
This holds because $\ln(x + h) - \ln x \le \frac{h}{x}$ for $x > 0$ and $h \ge 0$. Summing this over all $j \in \bc$ gives 
$$
\sum_{j \in \bc} F_i(\bc)-F_i(\bc \setminus \{j\}) \le \sum_{j \in \bc} \frac{u_{ij}}{u_i(\bc)} = \frac{u_i(\bc)}{u_i(\bc)} = 1.
$$
For an agent $i$ with $u_i(\bc) = 0$, we trivially have $\sum_{j \in \bc} F_i(\bc)-F_i(\bc \setminus \{j\}) = 0$. Summing over all agents, we have that for every outcome $\vec{c}$:
\begin{equation}
\label{eq1}\sum_{j \in \bc} F(\bc) - F(\bc \setminus \{j\}) = \sum_{i \in N} \sum_{j \in \bc} F_i(\bc) - F_i(\bc \setminus \{j\}) \le n.
\end{equation}

We now use the following result: 
\begin{lemma}[\cite{Kung}] For every pair of bases $\bc$ and $\bcp$ of a matroid $\M$, there is a bijection $f: \bc \rightarrow \bcp$ such that for every $j \in \bc$, $\bc \setminus \{j\} \cup \{f(j)\}$ is also a basis.
\end{lemma} 

Using the above lemma, combined with the fact that $\ln(x+h) - \ln x \ge \frac{h}{x+h}$ for $x > 0$ and $h \ge 0$, we have that for all $\bc,\bcp$:
\begin{align}
\sum_{j \in \bc}  F(\bc \setminus \{j\} \cup \{f(j)\}) - F(\bc \setminus \{j\})  & \ge \sum_i \sum_{j \in \bc}  \frac{u_{if(j)}}{u_i(\bc) + 1 - u_{ij} + u_{if(j)}}\nonumber\\
& \ge \sum_{i \in S} \sum_{j' \in \bcp}  \frac{u_{ij'}}{u_i(\bc) + 2}  = \sum_{i \in S} \frac{u_i(\bcp)}{u_i(\bc) + 2}.\label{eq2}
\end{align}

We now provide almost similar proofs for the approximations achieved by the global optimum $\bcstar$ and the local optimum $\bchat$.

{\em Global optimum.} Suppose for contradiction that $\bcstar$ is not a $(0,2)$-core outcome. Then, there exist a subset $S$ of agents and an outcome $\bcp$ such that for all $i \in S$,
$$
\frac{|S|}{n} \cdot u_i(\bcp) \ge u_i(\bcstar) + 2,
$$
and at least one inequality is strict. Rearranging the terms and summing over all $i \in S$, we obtain:
$$
\sum_{i \in S} \frac{u_i(\bcp)}{u_i(\bcstar) + 2} > \sum_{i \in S} \frac{n}{|S|} = n.
\label{eqn:global-opt}
$$

Combining this with Equation~\eqref{eq2}, and subtracting Equation~\eqref{eq1} yields:
$$
\sum_{j \in \bcstar}  \left( F(\bcstar \setminus \{j\} \cup \{f(j)\}) - F(\bcstar)  \right) >0.
$$
This implies existence of a pair $(j,f(j))$ such that $F(\bcstar \setminus \{j\} \cup \{f(j)\}) - F(\bcstar) > 0$, which contradicts the optimality of $\bcstar$ because $\bcstar \setminus \{j\} \cup \{f(j)\}$ is also a basis of $\M$.

{\em Local optimum.} Similarly, suppose for contradiction that $\bchat$ is not a $(0,2+\epsilon)$-core outcome. Then, there exist a subset $S$ of agents and an outcome $\bcp$ such that for all $i \in S$,
$$
\frac{|S|}{n} \cdot u_i(\bcp) \ge u_i(\bchat) + 2 + \epsilon > (1+\gamma) \left(u_i(\bchat) + 2 \right).
$$
Here, the final transition holds because $\gamma < \epsilon/(m+2) \le \epsilon/(u_i(\bchat)+2)$. 
Again, rearranging and summing over all $i \in S$, we obtain:
$$
\sum_{i \in S} \frac{u_i(\bcp)}{u_i(\bchat) + 2} > (1+\gamma) \sum_{i \in S} \frac{n}{|S|} \ge n \cdot (1+\gamma).
$$

Once again, combining this with Equation~\eqref{eq2}, and subtracting Equation~\eqref{eq1} yields:
$$
\sum_{j \in \bchat}  \left( F(\bchat \setminus \{j\} \cup \{f(j)\}) - F(\bchat)  \right)> n \gamma.
$$
This implies existence of a pair $(j,f(j))$ such that $F(\bchat \setminus \{j\} \cup \{f(j)\}) - F(\bchat) >  n \gamma/m$, which violates local optimality of $\bchat$ because $\bchat \setminus \{j\} \cup \{f(j)\}$ is also a basis of $\M$.

\paragraph{Lower Bound.} While a $(0,2)$-core always outcome exists, we show in the following example that a $(0,1-\epsilon)$-core outcome is not guaranteed to exist for any $\epsilon > 0$. 
\begin{lemma}
\label{lem:noexist}
For $\epsilon > 0$ and matroid constraints, $(0, 1-\epsilon)$-core outcomes are not guaranteed to exist.
\end{lemma}
\begin{proof} 
Consider the following instance of public decision making where we have several issues and must choose a single alternative for each issue, a special case of matroid constraints. There are $n$ agents, where $n$ is even. There are $m = (n-2)+n/2$ issues. The first $n-2$ issues correspond to unit-value private goods, \textit{i.e.}, each such issue has $n$ alternatives, and each alternative gives utility $1$ to a unique agent and utility $0$ to others. The remaining $n/2$ issues are ``pair issues''; each such issue has $\binom{n}{2}$ alternatives, one corresponding to every pair of agents that gives both agents in the pair utility $1$ and all other agents utility $0$.
	
It is easy to see that every integer allocation gives utility at most $1$ to at least two agents. Consider the deviating coalition consisting of these two agents. They can choose the alternative that gives them each utility $1$ on every pair issue, and split the $n-2$ private goods equally. Thus, they each get utility $n/2 + (n-2)/2 = n-1$. For the outcome to be a $(0,\alpha)$-core outcome, we need $1 + \alpha \ge (2/n) \cdot (n-1)$. As $n \to \infty$, this requires $\alpha \rightarrow 1$. Hence, for any $\epsilon > 0$, a $(0,1-\epsilon)$-core outcome is not guaranteed to exist. 
\end{proof}

Note that Theorem~\ref{thm:matroid} shows existence of a $(0,2)$-core outcome, which is therefore tight up to a unit additive relaxation. Whether a $(0,1)$-core outcome always exists under matroid constraints remains an important open question. Interestingly, we show that such an outcome always exists for the special case of private goods allocation, and, in fact, can be achieved by maximizing the smooth Nash welfare objective. 

\begin{lemma}
\label{lem:privatecore}
For private goods allocation, maximizing the smooth Nash welfare objective with $\ell = 1$ returns a $(0,1)$-core outcome.
\end{lemma}
\begin{proof}
There is a set of agents $N$ and a set of private goods $M$. Each agent $i \in N$ has a utility function $u_i : 2^M \to \bbR_{\ge 0}$. Utilities are additive, so $u_i(S) = \sum_{g \in S} u_i(\set{g})$ for all $S \subseteq M$. For simplicity, we denote $u_{ig} \triangleq u_i(\set{g})$. Without loss of generality, we normalize the utility of each agent such that $\max_{g \in M} u_{ig} = 1$ for each $i$. An allocation $A$ is a {\em partition} of the set of goods among the agents; let $A_i$ denote the bundle of goods received by agent $i$. We want to show that an allocation maximizing the objective $\prod_{i \in N} (1+u_i(A_i))$ is a $(0,1)$-core outcome. 

Let $A$ denote an allocation maximizing the smooth Nash welfare objective with $\ell = 1$. We assume without loss of generality that every good is positively valued by at least one agent. Hence, $u_j(A_j) = 0$ must imply $A_j = \emptyset$. 

For agents $i,j \in N$ with $A_j \neq \emptyset$ (hence $u_j(A_j) > 0$), and good $g \in A_j$, moving $g$ to $A_i$ should not increase the objective function. Hence, for each $g \in A_j$, we have 
$$
\big(1+u_i(A_i \cup \set{g})\big) \cdot \big(1+u_j(A_j \setminus \set{g})\big) \le \big(1+u_i(A_i)\big) \cdot \big(1+u_j(A_j)\big).
$$

Using additivity of utilities, this simplifies to
\begin{equation}
\frac{u_{ig}}{1 + u_i(A_i)} \le \frac{u_{jg}}{1+u_j(A_j)-u_{jg}} \le \frac{u_{jg}}{u_j(A_j)}.
\label{eqn:privatecore1}
\end{equation}

For every agent $j \in N$ with $A_j \neq \emptyset$ and good $g \in A_j$, define $p_g = u_{jg}/u_j(A_j)$. Abusing the notation a little, for a set $T \subseteq M$ define $p_T = \sum_{g \in T} p_g$. Then, from Equation~\eqref{eqn:privatecore1}, we have that for all players $i \in N$ and goods $g \in M$, 
\begin{equation}
(1+u_i(A_i)) \cdot p_g \ge u_{ig}.
\label{eqn:privatecore2}
\end{equation}

Suppose for contradiction that $A$ is not a $(0,1)$-core outcome. Then, there exists a set of agents $S \subseteq N$ and an allocation $B$ of the set of all goods to agents in $S$ such that $(|S|/n) \cdot u_i(B_i) \ge 1+u_i(A_i)$ for every agent $i \in S$, and at least one inequality is strict. Rearranging the terms and summing over $i \in S$, we have
\begin{equation}
\sum_{i \in S} \frac{u_i(B_i)}{1+u_i(A_i)} > \sum_{i \in S} \frac{n}{|S|} = n.
\label{eqn:privatecore3}
\end{equation}

We now derive a contradiction. For agent $i \in S$, summing Equation~\eqref{eqn:privatecore2} over $g \in B_i$, we get
$$
(1+u_i(A_i)) \cdot p_{B_i} \ge u_i(B_i) \Rightarrow \frac{u_i(B_i)}{1+u_i(A_i)} \le p_{B_i}.
$$
Summing this over $i \in S$, we get 
$$
\sum_{i \in S} \frac{u_i(B_i)}{1+u_i(A_i)} \le \sum_{i \in S} p_{B_i} = \sum_{g \in M} p_g = \sum_{\substack{j \in N \text{ s.t.}\\A_j \neq \emptyset}} \sum_{g \in A_j} \frac{u_{jg}}{u_j(A_j)} = \sum_{\substack{j \in N \text{ s.t.}\\A_j \neq \emptyset}} \frac{u_j(A_j)}{u_j(A_j)} \le n.
$$
However, this contradicts Equation~\eqref{eqn:privatecore3}.	
\end{proof}

\section{Matching Constraints}
\label{sec:matching}
We now present the algorithm proving Theorem~\ref{thm:matching}.  We show that if the elements are edges of an undirected graph $G(V,E)$, and the feasibility constraints encode a matching, then for constant $\delta \in (0,1]$,  a $(\delta, 8 + \frac{6}{\delta})$-core always exists and is efficiently computable. The idea is to again run a {\em local search} on the smooth Nash welfare objective in Equation~(\ref{eq:FK}), but this time with $\ell \approx 1+\frac{4}{\delta}$.

\paragraph{Algorithm.} Recall that there are $n$ agents. Let $|V| = r$ and $|E| = m$. Let $\kappa = \frac{2}{\delta}$. For simplicity, assume $\kappa \in \mathbb{N}$. Our algorithm is inspired by the PRAM algorithm for approximate maximum weight matchings due to~\cite{HV}, and we follow their terminology. Given a matching $\bc$, an {\em augmentation} with respect to $\bc$ is a matching $T \subseteq E \setminus \bc$. The {\em size} of the augmentation is $|T|$. Let $M(T)$ denote the subset of edges of $\bc$ that have a vertex which is matched under $T$. Then, the matching $\left(\bc\setminus M(T)\right) \cup T$ is called the {\em augmentation} of $\bc$ using $T$.

\begin{enumerate}
\item Start with an arbitrary matching $\bc$ of $G$.
\item Compute $F(\bc) = \sum_i \ln \left(1 + 2 \kappa + u_i(\bc) \right)$. 
\item Let $\calC$ be the set of all augmentations with respect to $\bc$ of size at most $\kappa$. 
\begin{itemize}
\item If there exists $T \in \calC$ such that $F(\left(\bc\setminus M(T)\right) \cup T) - F(\bc) \ge \frac{n}{\kappa r}$, perform this augmentation (i.e., let $\bc \gets \left(\bc\setminus M(T)\right) \cup T$) and go to Step (2).
\item Otherwise, output $\bc$ as the final outcome.
\end{itemize}
\end{enumerate}

\paragraph{Analysis.} 


The outline of the analysis is similar to the analysis for matroid constraints. First, we show that the algorithm 
runs in polynomial time. Again, recall that each agent has utility at most $m$. Thus, $F(\bc) = O(n \cdot \ln m)$. Because each improvement increases the objective value by at least $n/(\kappa r)$, the number of iterations is $O(\kappa r \ln m) = O(m^2/\delta)$. Each iteration can be implemented by na\"{\i}vely going over all $O(m^{\kappa})$ subsets of edges of size at most $\kappa$, checking if they are valid augmentations with respect to $\bc$, and whether they improve the objective function by more than $n/(\kappa r)$. The local search therefore runs in polynomial time for constant $\delta > 0$.

Let $\bc$ denote the outcome returned by the algorithm. We next show that $\bc$ is indeed a $(\delta,8+3\kappa)$-core outcome. Suppose for contradiction that this is not true. Then, there exist a subset of agents $S$ and a matching $\bcp$ such that for all $i \in S$,
$$
\frac{|S|}{n} \cdot u_i(\bcp) \ge (1+\delta) \cdot u_i(\bc) + 8 + 3 \kappa \ge (1+ \delta) \cdot \left(u_i(\bc) + 3 \kappa + 1 \right),
$$
and at least one inequality is strict (the last inequality is because $\delta \in (0,1]$). Rearranging and summing over all $i \in S$, we obtain
\begin{equation}
\sum_{i \in S} \frac{u_i(\bcp)}{u_i(\bc) + 3 \kappa + 1} > (1+\delta) \cdot \sum_{i \in S} \frac{n}{|S|} = n \cdot (1+\delta).
\label{eqn:matching-inequality1}
\end{equation}

For $j \in E$, define  $w_j = \sum_{i \in N} \frac{u_{ij}}{u_i(\bc) + 1}$ and  $w'_j =  \sum_{i \in N} \frac{u_{ij}}{u_i(\bc) + 3 \kappa + 1}$. Let $W = \sum_{j \in \bc} w_j$, and $W' = \sum_{j \in \bcp} w'_j$.  It is easy to check that 
\begin{equation}
W \le n \quad \text{and} \quad W' \ge n \cdot (1+\delta),
\label{eqn:matching-inequality2}
\end{equation}
where the latter follows from Equation~\eqref{eqn:matching-inequality1}.  Further note that $w_j \ge w'_j$ for all $j$.

For an augmentation $T$ with respect to $\bc$, define $\gain(T) = \sum_{j \in T} w'_j - \sum_{j \in M(T)} w_j$.  The next lemma is a simple generalization of the analysis in~\cite{HV}; we give the adaptation here for completeness.

\begin{lemma}  Assuming weights $w_j \ge w'_j$ for all edges $j$, for any integer $\kappa \ge 1$ and matchings $\bc$ and $\bcp$, there exists a multiset $OPT$ of augmentations with respect to $\bc$ such that:
\begin{itemize}
\item For each $T \in OPT$, $T \subseteq \bcp$ and $|T| \le \kappa$; 
\item $|OPT| \le \kappa r$; and
\item $\sum_{T \in OPT} \gain(T) \ge \kappa \cdot W' - (\kappa + 1) \cdot W$.
\end{itemize}
\label{lem:hv}
\end{lemma}
\begin{proof}
We follow~\cite{HV} in the construction the multiset $OPT$ of augmentations with respect to $\bc$ out of edges in $\bcp$. Let $\bc \triangle \bcp$ be the symmetric difference of matchings $\bc$ and $\bcp$ consisting of alternating paths and cycles. For every cycle or path $\bd \in \bc \triangle \bcp$, let $T_{\bd}$ be be set of edges $\bd \cap \bcp$. For all $T_{\bd}$ with $|T_{\bd}| \leq \kappa$, just add $T_S$ to OPT $\kappa$ times (note that $OPT$ is a multiset, not a set). For $T_{\bd}$ with $|T_{\bd}| > \kappa$, we break up $T_{\bd}$ into multiple smaller augmentations. To do so, index the edges in $T_{\bd}$ from $1$ to $|T_{\bd}|$ and add $|T_{\bd}|$ different augmentations to $OPT$ by considering starting at every index in $T_{\bd}$ and including the next $\kappa$ edges in $T_{\bd}$ with wrap-around from $|T_{\bd}|$ to $1$. 
	
Now we must argue that $OPT$ as we have constructed it satisfies the conditions of the lemma. The first point, that $\forall T \in OPT, T \subseteq \bcp$ and $|T| \leq \kappa$, follows trivially from the construction. The second point also follows easily from observing that we add $\kappa$ augmentations to $OPT$ for every $\bd \in \bc \cap \bcp$, and graph $G$ has $r$ vertices. 
	
To see the third point, note that every edge in  $\bcp \backslash \bc$ is contained in at least $\kappa$ augmentations in $OPT$. On the other hand, for every edge $e \in \bc \backslash \bcp$, there are no more than $\kappa + 1$ augmentations $T \in OPT$ such that $e \in M(T)$ (recall $M(T)$ are the edges of $\bc$ with a vertex matched under $T$). This can happen, for example, if $T_S$ happens to be a path of length $\kappa + 1$. Finally, for the edges $j \in \bcp \cap \bc$, the weight $w'_j \le w_j$.  Putting these facts together, the third point of the lemma follows. 
\end{proof}

Consider the set of augmentations $OPT$ from Lemma~\ref{lem:hv}. For augmentation $T \in OPT$, we have:

\begin{align*}
 F(\left(\bc\setminus M(T)\right) \cup T) - F(\bc) & = \Big( F(\left(\bc\setminus M(T)\right) \cup T) - F(\bc \setminus M(T)) \Big) - \Big(F(\bc) - F(\bc \setminus M(T)) \Big)\\
 & \ge \sum_{i \in N} \left(  \frac{ \sum_{T \in S} u_{ij}}{u_i(\bc) + 2 \kappa + 1 + \sum_{j \in T} u_{ij}} -  \frac{ \sum_{j \in M(T)} u_{ij}}{u_i(\bc) + 2 \kappa + 1 - \sum_{j \in M(T)} u_{ij}} \right) \\
 & \ge \sum_{i\in N} \left(  \frac{ \sum_{j \in T} u_{ij}}{u_i(\bc) + 3 \kappa + 1} -  \frac{ \sum_{j \in M(T)} u_{ij}}{u_i(\bc)+ 1} \right) \\
 & = \sum_{j \in T} w'_j - \sum_{j \in M(T)} w_j  = \gain(T).
 \end{align*}
Here, the second transition holds because $h/(x+h) \le \ln(x+h)-\ln x \le h/x$ for all $x \ge 1$ and $h \ge 0$, and the third transition holds due to $|T| \le \kappa$ and $|M(T)| \le 2|T| \le 2\kappa$. Therefore, we have:
\begin{align*}
\sum_{T \in OPT}  F(\left(\bc\setminus M(T)\right) \cup T) - F(\bc)  
\ge \sum_{T \in OPT} \gain(T) 
&\ge \kappa \cdot W' - (\kappa + 1) \cdot W\\
&\ge \kappa \cdot n \cdot (1+\delta) - (\kappa + 1) \cdot n 
= n,
\end{align*}
where the second transition follows from Lemma~\ref{lem:hv}, and the third transition follows from Equation~\eqref{eqn:matching-inequality2}. Since $|OPT| \le \kappa r$, there exists an augmentation $T \in OPT$ with $F(\left(\bc\setminus M(T)\right) \cup T) - F(\bc) \ge \sfrac{n}{\kappa r}$, which violates local optimality of $\bc$. This completes the proof of Theorem~\ref{thm:matching}.

\paragraph{Lower Bound.} We give a stronger lower bound for matchings than the lower bound for matroids in Lemma~\ref{lem:noexist}.

\begin{lemma}
A $(\delta, \alpha)$-core outcome is not guaranteed to exist for matching constraints, for any $\delta \ge 0$ and $\alpha < 1$.
\end{lemma}
\begin{proof}
This example shows that  Consider the graph $K_{2,2}$ (the complete bipartite graph with two vertices on each side). This graph has four edges, and two disjoint perfect matchings. 
	
Let there be two agents. Agent $1$ has unit utility for the edges of one matching, while agent $2$ has unit utility for the edges of the other matching. Any integer outcome gives zero utility to one of these agents. This agent can deviate and obtain utility $2$. Hence, for an outcome to be a $(\delta,\alpha)$-core outcome, we need $(1+\delta)\cdot 0 + \alpha \ge (1/2) \cdot 2$, which is impossible for any $\delta \ge 0$ and $\alpha < 1$.
\end{proof}
\section{General Packing Constraints}
\label{sec:main}
In this section, we study approximation to the core under general packing constraints of the form $A \vec{x} \le \vec{b}$. Recall that there are $m$ elements, $V_i$ is the maximum possible utility that agent $i$ can receive from a feasible outcome, and $V_{\max} = \max_{i \in N} V_i$.  We prove a statement slightly more general than Theorem~\ref{thm:packing}. We first need the following concept.

\subsection{Maximal Proportionally Fair Outcome}
\label{sec:mmf}
Given an instance of public goods allocation subject to packing constraints, we define the notions of an $r$-proportionally fair ($r$-PF) outcome, a maximally proportionally fair (MPF) outcome, and the MPF value of the instance.

\begin{definition}[MPF Outcome]
	\label{def:rmmf}
	For $r > 0$, we say that a fractional outcome $\vec{w}$ is $r$-\emph{proportionally fair} ($r$-PF) if it satisfies:
	$$
	u_i(\mathbf{w}) \ge \frac{V_i}{r} - 1,  \ \ \forall i \in N.
	$$
	The {\em maximally proportionally fair} (MPF) value $R$ of an instance is the least value $r$ such that there exists an $r$-PF outcome. For simplicity, we say that an $R$-PF outcome is a {\em maximally proportionally fair} (MPF) outcome.
\end{definition}

This concept is crucial to stating and deriving our approximation results. In words, an $r$-PF outcome gives each agent an $r$ fraction of its maximum possible utility $V_i$ (which can be thought of as the fair share guarantee of the agent), if the agent is given $1$ unit of utility for free. Thus, a smaller value of $r$ indicates a better solution. The MPF value $R$ denotes the best possible guarantee. 
The additive $1$ in Def.~\ref{def:rmmf} can be replaced by any positive constant; we choose $1$ for simplicity.


We now show an upper bound for $R$ that holds for all instances. Recall from Equation~(\ref{eq:width}) that $\rho$ is the {\em width} of the instance.
\begin{lemma}
$R \le \min(V_{\max},n,\rho)$, and an MPF outcome is computable in polynomial time.
\label{lem:R-MMF}
\end{lemma}
\begin{proof}
To show that $R$ is well-defined, note that for $r = V_{\max}$, an $r$-PF outcome $\vec{w}$ simply requires $u_i(\vec{w}) \ge 0$, which is trivially achieved by every outcome. Therefore, $R$ is well-defined, and $R \le V_{\max}$.  Next, $R \le n$ follows from the fact that there exist fractional outcomes satisfying proportionality (e.g., the outcome $\vec{w}$ obtained by taking the uniform convex combination of the $n$ outcomes that give optimal to each individual agent). Finally, to show $R \le \rho$, consider the outcome $\vec{w}$ in which $w_j = \frac{1}{\rho}$ for each element $j$. Clearly, $u_i(\vec{w}) \ge \frac{V_i}{\rho}$ for all $i$.  Further, $A \vec{w} \le \vec{b}$ is satisfied trivially due to the fact that $\rho$ is the width of the packing constraints.

To compute the value of $R$ as well as an MPF outcome, we first note that the value of $V_i$ for each agent $i$ can be computed by solving a separate LP. Then, we consider the following LP:
\begin{equation}
\label{eq:compute-R} 
\mbox{Maximize} \ \ \hat{r}
\end{equation}
\[ \begin{array}{rcll}
\sum_{j \in W} u_{ij} w_j  & \ge & V_i \cdot \hat{r} - 1 & \forall i \in [n] \\
A \vec{w} & \leq  & \vec{b} & \\
w_j & \in &  [0,1] &  \forall j \in W
\end{array}\]
Here, $A \vec{w} \le \vec{b}$ are the packing constraints of the instance, and $\hat{r}$ is a variable representing $1/r$. Thus, maximizing $\hat{r}$ minimizes $r$, which yields an MPF outcome. This can be accomplished by solving $n+1$ linear programs, which can be done in polynomial time.
\end{proof}

Our main result in this section uses any $r$-PF outcome, and provides a guarantee in terms of $\log r$. Thus, we do not need to necessarily compute an exact MPF outcome. We note that an MPF outcome can be very different from a core outcome. Yet, an MPF outcome gives each agent a large fraction of its maximum possible utility, subject to a small additive relaxation. As we show below, this helps us find integral outcomes that provide good approximations of the core.

\subsection{Result and Proof Idea}
\label{sec:idea}
Our main result for this section (Theorem~\ref{thm:packing}) can be stated in a refined way as follows. Recall that $\log^*$ is the iterated logarithm, which is the number of times the logarithm function must be iteratively applied before the result becomes less than or equal to $1$.

\begin{theorem}
\label{thm:main}
 Fix constant $\delta \in (0,1)$. Suppose we are given a set of $K$ packing constraints $A \vec{x} \le \vec{b}$  such that $b_k = \omega\left(\frac{\log K}{\delta^2} \right)$ for all $k \in [K]$. Let $R$ be the MPF value of this instance. Then there exists a polynomial time computable $(\delta,\alpha)$-core outcome, where 
$$
\alpha = O\left( \frac{1}{\delta^4}  \cdot \log\left(\frac{R \cdot \log^* V_{\max}}{\delta} \right) \right).
$$
\end{theorem}

We first note that the above result cannot be obtained by maximizing the smooth Nash welfare objective; we present Example~\ref{eg:knapsack} in Appendix~\ref{sec:Examples}, which demonstrates this using only one packing constraint. To be precise, the example shows that no single value of parameter $\ell$ in the smooth Nash welfare objective can provide a polylog additive guarantee for all instances. While it may be possible to choose the value of $\ell$ based on the instance, it does not seem trivial. We take a different approach. Our idea is to start with a fractional core solution $\vec{x}$. Suppose it assigns utility $U^*_i$ to agent $i$. Fix $\delta > 0$, and consider the following program.

\begin{equation}
\label{eq:round} 
\mbox{Minimize} \ \ \alpha
\end{equation}
\[ \begin{array}{rcll}
\alpha + (1+\delta) \cdot \sum_{j \in W} u_{ij} w_j  & \ge & U^*_i  & \forall i \in [n] \\
A \vec{w} & \leq  & \vec{b} & \\
 w_j & \in &  \{0,1\} &  \forall j \in W  \\
 \alpha & \ge & 0
\end{array}\]

For the optimum value $\alpha^*$, we obtain an outcome that is $(\delta,\alpha')$-core for every $\alpha' > \alpha$. To see this, take a subset of agents $S$ and a feasible utility vector $\vec{U'}$ under any other (even fractional) outcome. Because $\vec{x}$ is a core outcome, there exists $i \in S$ such that $U^*_i \ge (|S|/n) \cdot U'_i$. For $\alpha' > \alpha$, the ILP solution implies 
$$
\alpha' + (1+\delta) \cdot \sum_{j \in W} u_{ij} w_j  >  U^*_i  \ge  \frac{|S|}{n} \cdot U'_i,
$$ 
which implies that the solution is $(\delta,\alpha')$-core according to Definition~\ref{def:approx}. 

However, $\alpha^*$ obtained from this program can be rather large, as illustrated in the following example. Consider the {\sc Knapsack} setting with $m$ unit-size projects. There is an overall budget $B = m/2$. For every feasible integral outcome $\bc$, let there be an agent with utility $1$ for every project in $\bc$ and $0$ for all other projects. Thus, there are $\binom{m}{m/2}$ agents. The fractional core outcome gives weight $1/2$ to each project, thus giving utility $V_i / 2 = m / 4$ to each agent $i$. However, every integral outcome gives utility $0$ to at least one agent, which implies $\alpha^* = \Omega(m)$.

This example shows that when there are a large number of agents, we cannot achieve Theorem~\ref{thm:main} by hoping to approximately preserve the utilities to {\em all} agents with respect the fractional core solution. 
However, note that in the above example, though there is one agent who gets very little utility, this agent has no incentive to deviate if she is given one unit of utility for free. This insight leads us to our analysis below, which is based on rounding the fractional core solution $\vec{x}$. 

Let us apply randomized rounding to $\vec{x}$. Instead of using Chernoff bounds to ensure that there are no ``violations'' (i.e., that no agent receives utility that is too far from its utility under the core outcome $\vec{x}$), we hope to bound the expected number of such violations. If there are few such agents, we still have an approximate core outcome because if this small coalition of agents deviates, its utility under a new outcome will be scaled down by a large factor. Unfortunately, it can be shown that bounding the expected number of deviations by a sufficiently small number forces $\alpha = \Omega(\log V_{\max})$. This is better than $\alpha = \Omega(m)$ from our previous approach, but still {\em much} larger than the bound we want to achieve in Theorem~\ref{thm:main} when the width $\rho$ is small.

This brings up the main technical idea. We observe that an MPF outcome, though not in the core, provides a reasonably large utility to each agent. We add a small amount of this outcome to the fractional core before applying randomized rounding. We are now ready to present our algorithm.

\subsection{Algorithm}
Fix $\delta \in \left(0,1\right)$, and let  $\gamma = \frac{\delta}{8}$. 

\begin{enumerate}
\item Compute the (approximate) fractional core solution $\vec{x}$ as in Theorem~\ref{thm:fractional}, where $x_j$ is the fraction of element $j$ chosen. 
\item Let $\vec{y}$ be an MPF outcome as in Definition~\ref{def:rmmf}. 
\item Let $\vec{z} =  (1-\gamma) \vec{x} + \gamma \vec{y}$.
\item For each $j \in W$, choose $j$ to be in the outcome $\bc$ independently with probability $\hat{z_j} = (1-\gamma) z_j$.
\end{enumerate}


\newcommand{\hl}{L}

\subsection{Analysis}
We show that this algorithm yields, with at least a constant probability, a feasible outcome that satisfies the guarantee in Theorem~\ref{thm:main}. This directly shows the existence of such an outcome. Note that the fractional Max Nash Welfare solution $\vec{x}$ can be irrational, but we can compute an approximation in polynomial time (see Theorem~\ref{thm:fractional} for details), which does not change our guarantee asymptotically. Further, $\vec{y}$ can be computed in polynomial time (Lemma~\ref{lem:R-MMF}). Hence, the algorithm runs in expected polynomial time. 
 
We first show that the packing constraints are satisfied. Since we scale down $\vec{z}$ by a factor $(1-\gamma)$ before rounding, we have $A \hat{z} \le  (1-\gamma) \vec{b}$. Since $\vec{b} = \omega\left(\frac{\log K}{\delta^2}\right)$, a simple application of Chernoff bounds shows that with probability at least $0.99$, the rounded solution $\bc$ satisfies $A \vec{\bc} \le \vec{b}$. Therefore, if we show that the algorithm also yields the desired approximation of the core with at least a constant probability ($1/6$ to be precise), we will be done by applying union bound to the two events, feasibility and good approximation to the core.

For the ease of presentation, we suppress constants throughout the proof and use the asymptotic notation liberally.  We also assume that $V_{\max} = \omega(1)$ since otherwise there is a trivial $(0,O(1))$-core outcome that chooses a null outcome, giving zero utility to each agent. 

\subsubsection{Grouping Agents} 
In order to analyze our algorithm, we partition the agents into groups with exponentially decreasing values of $V_i$. Recall that $V_i$ is the maximum utility that agent $i$ can get from any outcome. Set $Q_0 =  \log V_{\max}$, and for $\ell = 0,1,\ldots,\hl-1$, define group $G_{\ell}$ as:
$$
G_{\ell} = \left\{ i \in N \ | \ Q_{\ell} \ge \log V_i \ge Q_{\ell+1} \right\}.
$$
Here, for $\ell = 0,1,\ldots,\hl-1$, we define: $ Q_{\ell+1}  = 2 \log Q_{\ell}.$

We call $G_0,\ldots,G_{\hl-1}$ the {\em heavy groups}. We choose $\hl$ so that $Q_{\hl} = \Theta\left( \log \frac{R  \log^* V_{\max}}{\gamma^3} \right)$. This implies $L = \Omega(\log^* V_{\max}) = \omega(1)$, since $V_{\max} = \omega(1)$. For agent $i$ in a heavy group, $V_i \geq e^{Q_{\hl}} \ge \frac{2 R L}{\gamma^3} > 2 R$. Thus, the utility that the MPF solution provides to agent $i$ is at least $\frac{V_i}{R}-1 \ge \frac{V_i}{2 R}$.

Finally, we put the remaining agents (with a small $V_i$) in a {\em light group} defined as follows:
$$
G_{\hl} = \left\{ i \in N \ | \  \log V_i \le Q_{\hl}  \right\}.
$$
The MPF solution may not provide any guarantee for the utility of agents in this group. 

\subsubsection{Bounding Violations of Utility Preservation}
\label{sec:light}
We want to bound the number of agents whose utilities are far from those under the core outcome. First, we need a specialized Chernoff bound.
\begin{lemma} (Proved in Appendix~\ref{proof:main})
\label{lem:main}
Let $X_1, X_2, \ldots, X_q$ be independent  random variables in $[0,1]$, and let $X = \sum_{j=1}^q X_j$. For $\gamma \in (0,1/2)$, suppose $\E[X] = (1-\gamma) \cdot A + \gamma \cdot B$ for $A, B \geq 0 $. Then
$$ \Pr[X < (1-2\gamma) \cdot A] \le  e^{- \frac{\gamma^3}{2} \max (B,A/2)}$$
\end{lemma}

Recall that $\vec{x}$ is the fractional MNW solution, $\vec{y}$ is the fractional MPF solution, and our algorithm applies randomized rounding to their scaled down mixture $(1-\gamma) \vec{z} = (1-\gamma)^2 \vec{x} + \gamma(1-\gamma) \vec{y}$. Let $\widehat{U_i}$ denote the utility of agent $i$ under the final integral outcome obtained by randomly rounding $(1-\gamma) \vec{z}$. Recall that $U_i^*$ is the utility of agent $i$ under the core outcome $\vec{x}$. We want to show that $\widehat{U_i}$ is either multiplicatively or additively close to $U_i^*$ for most agents. For a heavy group $G_{\ell}$, where $\ell \in \set{0,1,\ldots,\hl-1}$, define 
$$
F_{\ell} = \set{i \in G_{\ell}\ \left|\ \widehat{U_i} < (1-3 \gamma) U^*_i \right.}.
$$

Simiarly, for the light group $G_{\hl}$, define 
$$
F_{\hl} = \set{i \in G_{\hl} \ \left| \ \widehat{U_i} < \min\left((1-3\gamma) U^*_i, U_i^* -  \frac{4 Q_{\hl}}{\gamma^4} \right)\right.}.
$$ 

We will use Lemma~\ref{lem:main} to bound the sizes of $F_{\ell}$ for $\ell \in \set{0,1,\ldots,\hl}$ as follows.
\begin{theorem}
\label{thm:main2}
We have that:
\begin{enumerate}
	\item With probability at least $2/3$, we have $|F_{\ell}| \le \frac{1}{2 L e^{Q_{\ell}}} \cdot |G_{\ell}|, \quad \forall \ell \in \set{0,1,\ldots,\hl-1}$.
	\item With probability at least $1/2$, we have $|F_{\hl}| \le \frac{1}{2 e^{Q_{\hl}}} \cdot |G_{\hl}|$.
\end{enumerate}
Thus, with probability at least $1/6$,  both the above inequalities hold simultaneously.
\end{theorem}
\begin{proof}
We prove the first and the second part of Theorem~\ref{thm:main2} separately by considering the heavy groups and the light group in turn. The combined result follows from the union bound.

\paragraph{Case 1: Heavy Groups} Consider a heavy group $G_{\ell}$ for $0 \le \ell < \hl$. 
Recall that the MPF solution provides utility at least $V_i/(2R)$ to each agent in a heavy group. Hence, we have: 
\begin{equation}
\E \left[ \widehat{U_i} / (1-\gamma) \right] = u_i(\vec{z}) \ge (1-\gamma) \cdot U^*_i + \gamma \cdot \frac{V_i}{2 R}.
\label{eqn:heavy-utility}
\end{equation}

The key point is that even if $U^*_i$ is small, the expected utility is at least a term that is proportional to $V_i$. This will strengthen our application of Chernoff bounds. Using Lemma~\ref{lem:main} with $A = U_i^*$ and $B = V_i/(2R)$, we have:
\begin{align}
\Pr\left[ \widehat{U_i} < (1- 3 \gamma) \cdot U^*_i \right] & \leq \Pr\left[ \frac{\widehat{U_i}}{1-\gamma} < (1- 2 \gamma) \cdot U^*_i \right] \nonumber\\ & \le e^{- \frac{\gamma^3}{4} \frac{V_i}{2R}}    \le  e^{- \frac{\gamma^3}{8R}  Q_{\ell}^2}\nonumber\\
& \le e^{- Q_{\ell} \cdot \log L}  \le e^{ - \left( Q_{\ell} + 2 \log L + \log 6 \right)}  \leq  \frac{1}{6 L^2 e^{Q_{\ell}}},\label{eqn:heavy-prob}
\end{align}
where the second inequality holds because $\log V_i \ge 2 \log Q_{\ell} $, the third holds because $Q_{\ell} = \Omega \left(\frac{R  L}{\gamma^3} \right)$, and the fourth holds because $L = \omega(1)$.

We are now ready to prove the first part of Theorem~\ref{thm:main2}. Let $\eta_{\ell} =   \frac{1}{6 L^2 e^{Q_{\ell}}}$. Recall that $F_{\ell}$ consists of agents in $G_{\ell}$ for which $\widehat{U_i}< (1-3\gamma) \cdot U^*_i$. Using the linearity of expectation in Equation~\eqref{eqn:heavy-prob}, we have $\E[F_{\ell}] \le  \eta_{\ell} \cdot |G_{\ell}|$. By Markov's inequality,  $
\Pr\left[ |F_{\ell}| > 3 L \cdot \eta_{\ell} \cdot |G_{\ell}| \right] \le \frac{1}{3L}.$
Applying the union bound over the $L$ heavy groups, we have that with probability at least $2/3$,
$$
|F_{\ell}| \le 3 L \cdot \eta_{\ell} \cdot |G_{\ell}| =  \frac{1}{ 2 L e^{Q_{\ell}}} \cdot |G_{\ell}|,\quad \forall \ell \in \{0,1,\ldots,L-1\},
$$
which proves the first part of Theorem~\ref{thm:main2}.

\paragraph{Case 2: Light Group}
For the light group, note that  $\log V_i \le Q_{\hl}$. For this group, the MPF solution may not provide any non-trivial guarantee on the utility to the agents. Since the expected utility can now be small, we have to allow additive approximation as well. Recall that $F_{\hl}$ consists of agents in $G_{\hl}$ for whom $\widehat{U_i} < (1-3\gamma) \cdot U^*_i$ {\em as well as} $\widehat{U_i} < U_i^* -  4Q_{\hl}/\gamma^4$. We again consider two cases. 

\medskip
\noindent {\bf Case 1.} If $U^*_i \le \frac{4}{\gamma^4} Q_{\hl}$, then  $\widehat{U}_i \ge U^*_i -  \frac{4 Q_{\hl}}{\gamma^4}$ trivially. 

\medskip
\noindent {\bf Case 2.} Otherwise, $U^*_i \ge \frac{4 Q_{\hl}}{\gamma^4}  $, and using Lemma~\ref{lem:main}, we have:
$$
\Pr\left[ \widehat{U}_i < (1- 3 \gamma) U^*_i \right]  \le \Pr\left[ \frac{\widehat{U}_i}{1-\gamma} < (1- 2 \gamma) U^*_i \right] \le e^{- \frac{\gamma^3}{4} U^*_i}     \le  e^{- \frac{\gamma^3}{4}  \cdot \frac{4 Q_{\hl}}{\gamma^4} }  \le  \frac{1}{4 e^{Q_{\hl}}}.
$$
It is easy to check that the final transition holds because $\gamma < 1$ is a constant and $Q_L = \omega(1)$. 

Note that none of the agents in $F_{\hl}$ are in Case 1. Hence, by Markov's inequality, we again have:
$$
\Pr\left[|F_{\hl}| \ge \frac{1}{2e^{Q_{\hl}}} \cdot |G_{\hl}| \right] \le \frac{1}{2},
$$
which proves the second part of Theorem~\ref{thm:main2}. 

\end{proof}

\subsubsection{Approximate Core}
We showed that with probability at least $1/6$, our algorithm returns a solution that satisfies conditions in both parts of Theorem~\ref{thm:main2}. We now show that such a solution is the desired approximate core solution. The main idea is that when a set of agents deviate, the fraction of agents in a group $G_{\ell}$ that are in $F_{\ell}$ is small enough such that even if they receive their maximum possible utility, which is $e^{Q_{\ell}}$, their scaled down utility is at most a constant.

\begin{theorem} 
\label{thm:main3}
For every coalition $S$ and every possible outcome $\vec{h}$, there exists an agent $i \in S$ s.t. 
$$
\frac{|S|}{n} \cdot u_i(\vec{h}) \le (1+8\gamma) \cdot \widehat{U_i} + \frac{5 Q_{\hl}}{\gamma^4}.
$$
\end{theorem}
\begin{proof}
Let $W = N \setminus \cup_{\ell = 0}^{L} F_{\ell}$. In other words, $W$ is the set of agents who either receive a good multiplicative approximation to their expected utility in the core (for the heavy groups), or a good additive approximation to their expected utility in the core (for the light group). In particular, for every $i \in W$, we have $\widehat{U_i} \ge \min\left((1-3 \gamma) \cdot U^*_i, U_i^* -  \frac{4 Q_{L}}{\gamma^4} \right)$, which implies
\begin{equation}
U^*_i \le \frac{1}{1-3 \gamma} \cdot \widehat{U_i} + \frac{4 Q_{L}}{\gamma^4}.\label{eqn:Ustar-U}
\end{equation}

Consider a set of agents $S$ that may want to deviate, and let $\vec{h}$ be any (even fractional) outcome. There are two cases:

\paragraph{Case 1.} Suppose $|S \cap W| \ge (1-\gamma) \cdot |S|$. Then, due to the fractional core optimality condition (see Section~\ref{sec:fractional}), we have: 
$$\sum_{i \in S \cap W} \frac{U_i(\vec{h})}{U^*_i} \le n.$$
Note that in polynomial time, Theorem~\ref{thm:fractional} only finds an approximate solution whose utilities $\{\tilde{U}_i\}$ with $\sum_{i \in S \cap W} \frac{U_i(\vec{h})}{\tilde{U}_i} \le n (1+\eta)$ for small $\eta > 0$. It is easy to check this does not alter the rest of the proof and adds a small multiplicative factor of $(1+\eta)$ to the final approximation bound. We ignore this factor for simplicity and simply assume $\{U_i^*\}$ are the optimal MNW utilities.
The above implies $
\frac{|S|}{n} \cdot  \sum_{i \in S \cap W}  \frac{u_i(\vec{h})}{U^*_i} \le |S| \le \frac{1}{1-\gamma} \cdot |S \cap W|$.

Therefore, there exists an agent $i \in S \cap W$ such that 
$$
\frac{|S|}{n} \cdot u_i(\vec{h}) \le \frac{1}{1-\gamma} \cdot U^*_i \le \frac{1}{1-\gamma} \cdot \left( \frac{1}{1-3\gamma} \cdot \widehat{U_i} + \frac{4 Q_{L}}{\gamma^4} \right),
$$
where the last transition is due to Equation~\eqref{eqn:Ustar-U} and the fact that $i \in W$. Finally, it is easy to check that for $\gamma = \delta/8 \le 1/8$, we have $\frac{1}{(1-\gamma)\cdot (1-3\gamma)} \le 1+8\gamma$ and $4/(1-\gamma) \le 5$, which yields:
\begin{equation}
\frac{|S|}{n} \cdot u_i(\vec{h}) \le (1 + 8 \gamma) \cdot \widehat{U_i} + \frac{5 Q_{\hl}}{\gamma^4}.
\label{eqn:main2-bound}
\end{equation}

\paragraph{Case 2.} Otherwise, $|S \setminus W| \ge \gamma |S|$. In this case, we want to show that there exists an agent $i \in S\setminus W$ such that $(|S|/n) \cdot u_i(\vec{h}) \le 1/\gamma$. Because $\widehat{U_i} \ge 0$ and $Q_{\hl} = \omega(1)$, such an agent will also satisfy Equation~\eqref{eqn:main2-bound}. We show this by taking two sub-cases. 

First, suppose the light group satisfies $|S \cap F_{\hl}| \ge \frac{\gamma}{2} |S|$. Then: $ |S| \le \frac{2}{\gamma} \cdot |S \cap F_{\hl}| \le \frac{2}{\gamma} \cdot |F_{\hl}|.$
Thus, for any agent $i \in F_L$, we have
$$
\frac{|S|}{n} \cdot u_i(\vec{h}) \le \frac{2}{\gamma n} \cdot |F_{\hl}| \cdot V_i \le \frac{2}{\gamma n} \cdot \frac{|G_{\hl}|}{2 e^{Q_{\hl}}} \cdot V_i \le \frac{1}{\gamma}.
$$
Here, the second transition follows from Theorem~\ref{thm:main2}. To see why the third transition holds, note that $|G_{\hl}| \le n$, and that $\log V_i \le Q_{\hl}$ because $i \in G_{\hl}$.
 
Similarly, in the other sub-case, suppose $|S \cap F_{\hl}| \le  \frac{\gamma}{2} |S|$. Then, there exists a heavy group $\ell \in \{0,1,\ldots,\hl-1\}$ such that $|S \cap F_{\ell}| \ge \frac{\gamma}{2L} |S|$. This means $
 |S| \le \frac{2L}{\gamma} \cdot |S \cap F_{\ell}| \le  \frac{2L}{\gamma} \cdot |F_{\ell}|.$

Again, for an arbitrary agent $i \in F_{\ell}$, we have:
$$
\frac{|S|}{n} \cdot u_i(\vec{h}) 
\le \frac{2 L}{\gamma n} \cdot |F_{\ell}| \cdot V_i 
\le \frac{2L}{\gamma n} \cdot \frac{|G_{\ell}|}{2 L e^{Q_{\ell}}} \cdot V_i \le \frac{1}{\gamma}.
$$
Once again, the third transition follows from Theorem~\ref{thm:main2}, and the fourth transition holds because $|G_{\ell}| \le n$ and $\log V_i \le Q_{\ell}$ as $i \in G_{\ell}$. Putting everything together, the theorem follows.
\end{proof}

Since $\gamma = \frac{\delta}{8}$ and $Q_{\hl} =   \Theta\left( \log \frac{R \log^* V_{max}}{\gamma} \right)$, Theorem~\ref{thm:main3} implies $
\frac{|S|}{n} \cdot u_i(\vec{h}) \le (1+\delta) \cdot \widehat{U_i} + \alpha^*$, where  $
\alpha^* = O\left(\frac{1}{\delta^4} \cdot \log\left(\frac{R \cdot \log^*V_{\max}}{\delta} \right) \right)$. 
The existence of such an agent implies that a solution satisfying Theorem~\ref{thm:main2} is a $(\delta,\alpha)$-core solution for every $\alpha > \alpha^*$, which completes the proof of Theorem~\ref{thm:main}.

\section{Conclusion}
We considered the problem of fairly allocating public goods. We argued that the {\em core}, which is a generalization of proportionality and Pareto efficiency, and approximations of the core provide reasonable fairness guarantees in this context. Given that no integral outcome may be in the core, we presented efficient algorithms to produce integral outcomes that are constant or near-constant approximations of the core, thereby also establishing the non-trivial existence of such outcomes. Note that our algorithms for matroid and matching constraints that globally optimize the smooth Nash welfare objective achieve {\em exact} rather than approximate Pareto efficiency, in addition to an approximation of the core. An interesting question is whether the same guarantee can be provided (regardless of computation time) for general packing constraints.

Another natural question following our work is to tighten our upper bounds, or to establish matching lower bounds. For instance, we show the existence of a $(0,2)$-core outcome for matroid constraints (Theorem~\ref{thm:matroid}), but our lower bound only shows that a $(0,1-\epsilon)$-core outcome may not exist. This leaves open the question of whether a $(0,1)$-core outcome always exists. Existence of $(0,1)$-core outcome is also an open question for matching constraints. For packing constraints, it is unknown if even a $(\delta,\alpha)$-core outcome exists for constant $\delta > 0$ and $\alpha = O(1)$. This also remains an open question for the endowment-based notion of core we consider in Appendix~\ref{sec:prop}.




At a higher level, we established connections between approximating the core in our multi-agent environment and the problem of finding the optimal (i.e., utility-maximizing) outcome for a {\em single} agent. For instance, given matching constraints, our algorithm uses the idea of {\em short} augmenting paths from fast PRAM algorithms. 
This hints at the possibility of a deeper connection between efficiency results and existence results.


\paragraph{Acknowledgment.} We thank Vincent Conitzer, Rupert Freeman, Ashish Goel, and Paul G\"{o}lz for helpful comments. Brandon Fain is supported by NSF grants CCF-1637397 and IIS-1447554. Kamesh Munagala is supported by NSF grants CCF-1408784, CCF-1637397, and IIS-1447554.




\bibliographystyle{acm}
\bibliography{abb,ultimate,refs}
\newpage
\appendix
\section*{Appendix}
\section{Impossibility Examples}
\label{sec:Examples}

\begin{example}
\label{eg:IS}
\normalfont{
The following example shows the necessity of assuming a large $\vec{b}$ in Theorem~\ref{thm:packing} for general packing constraints. Specifically, the example uses packing constraints $A \vec{x} \le \vec{b}$ where $\vec{b} = \vec{1}$ (and the width is $\rho = 2$), and does not admit a $(\delta, m/4)$-core outcome for any $\delta > 0$. 
	
Consider a complete bipartite graph $G(L,R,E)$, where $|L| = |R| = m/2$. The vertices are the elements of the ground set $W$, and the constraints ensure that feasible outcomes are independent sets. There are two agents.  Agent $1$ has unit utility for each vertex in $L$, and zero utility for each vertex in $R$, while agent $2$ has unit utility for each vertex in $R$. A feasible outcome is forced to choose either vertices from $L$ or vertices from $R$, and hence gives zero utility to at least one agent. But this agent can deviate and choose an outcome with utility $m/2$, which is then scaled down to $m/4$. Hence, no feasible outcome is $(\delta,m/4)$-core for any $\delta > 0$. 
	
Note that in this example, the welfare-maximizing outcome for a {\em single} agent is simple to compute, which shows that the non-existence of a good approximation to the core is orthogonal to the computational difficulty of the single-agent welfare maximization problem.}
\end{example}

\begin{example}
\label{eg:knapsack}
\normalfont{ 
Recall that in the {\sc Knapsack} setting, we are given a set of elements of different sizes, and our goal is to select a subset of elements with total size at most a given budget $B$. We show that for any $\ell > 0$, there exists a {\sc Knapsack} instance in which maximizing the smooth Nash welfare objective $F(\bc) = \sum_{i \in N} \ln{(\ell + u_i(\bc))}$ returns an outcome that is not a $(O(m^{1/2-\epsilon}),O(m^{3/4-\epsilon}))$-core outcome. This is in contrast to Theorem~\ref{thm:main}, which provides a $(\delta,\alpha)$-core guarantee where $\delta$ is constant and $\alpha$ is logarithmic in the number of elements. 
	
Fix $\ell > 0$. Set a large budget $B \geq \ell^4$. There are $m = B^{1/4}+B$ elements, of which $B^{1/4}$ are {\em large} elements of size $B^{3/4}$ and the remaining $B$ are {\em small} elements of unit size. There are $n \ge 4 B^{1/4} \log(2B)$ agents. Each agent has unit utility for each large element. A subset of $\alpha n$ agents are {\em special} (we determine $\alpha$ later). These special agents have unit utility for each small element, while the remaining agents have zero utility for the small elements. 
	
The idea is to show that when $\alpha$ is sufficiently small, the smooth Nash welfare objective will choose only the large elements. However, $\alpha$ can still be large enough so that the special agents can deviate, and get a large amount of utility. 
	
Note that maximizing the smooth Nash welfare objective returns a Pareto efficient solution, and hence can be one of two types: it either chooses all large elements (which gives utility $B^{1/4}$ to each agent), or it chooses $B^{1/4}-r$ large elements and $r B^{3/4}$ small elements. For the former to have a larger smooth Nash welfare objective value, we need that for each $1 \le r \le B^{1/4}$,
	$$
	\ln (B^{1/4} + \ell) > \alpha \ln \left(r B^{3/4} + (B^{1/4}-r) + \ell \right) + (1-\alpha) \ln(B^{1/4} - r + \ell).
	$$
	This holds true if 
	$$
	\ln\left(\frac{B^{1/4}+\ell}{B^{1/4}+\ell-r}\right) > \alpha \ln \left(r B^{3/4} + (B^{1/4}-r) + \ell \right).
	$$
	Since $0 < \ell \leq B^{1/4}$, the above is true for each $1 \le r \le B^{1/4}$ if
	$$
	\ln \left(\frac{2B^{1/4}}{2B^{1/4}-1}\right) > \alpha \ln \left(B^{1/4} B^{3/4} + B^{1/4} \right).
	$$
	This is true when 
	$$
	\ln \left(1 + \frac{1}{2B^{1/4}}\right) \ge \alpha \ln (2B).
	$$
	
	Since $\ln(1+x) \ge x/2$ for $x \in [0,1]$, the above holds when:
	$$
	\alpha \le \frac{1}{4B^{1/4} \ln (2B)}.
	$$
	
	Let us set $\alpha = \frac{1}{4B^{1/4} \ln (2B)}$. Choosing all large elements maximizes the smooth Nash welfare objective. Since $n \ge 1/\alpha$, there is at least one special agent. The special agents get utility $B^{1/4}$ each. If they deviate and choose all the small elements, they get (scaled down) utility 
	$$
	\alpha B = \frac{B}{4B^{1/4}\ln (2B)} = \frac{B^{3/4}}{4\ln(2B)}.
	$$
	
	Hence, for the solution to be a $(\delta,\alpha)$-core outcome, we need $(1+\delta) \cdot B^{1/4} + \alpha \ge B^{3/4}/(4\ln (2B))$. Since $m = \Theta(B)$, this shows that the outcome is not a $(O(m^{1/2-\epsilon}),O(m^{3/4-\epsilon}))$-core outcome for any constant $\epsilon > 0$, as required. }
\end{example}

\section{Omitted Proofs}
\label{sec:omittedProof}

\subsection{Proof of Theorem~\ref{thm:fractional}}
\label{app:fractional}
The fractional outcome maximizing the Nash welfare objective is the solution of the following program. For simplicity of presentation, we absorb the constraint that $w_j \le 1$ for each $j \in W$ into the packing constraints.

\begin{equation}
\label{prog:2}
\mbox{Maximize } \ \sum_{i \in N} \ln U_i 
\end{equation}
\[ \begin{array}{rcll}
\sum_{j=1}^{m} a_{kj}  w_j  & \leq  & b_k & \forall k \in [K] \\
U_i & = &   \sum_{j =1}^m  w_j u_{ij}  & \forall i \in N \\
w_j  & \ge &  0 &  \forall j \in W 
\end{array}\]

\renewcommand{\P}{\mathcal{U}}

Denote a vector of utilities by $\vec{U} = \langle U_1,U_2, \ldots, U_n \rangle$, and the polytope of feasible utility vectors by $\P$. Then, the fractional MNW outcome is obtained by the following maximization. 
\begin{equation}
\label{prog:3}
 \max_{\vec{U} \in \P} \ \sum_{i \in N} \ln U_i 
 \end{equation}

We want to compute a fractional $(\delta,\epsilon)$-approximate core outcome in time polynomial in $n, V_{\max},$ and $\log \frac{1}{\delta\epsilon}$. 
Assume that $\P$ is a convex polytope of feasible utility vectors. For any $\delta \geq 0, \epsilon > 0$, let $\epsilon' = \epsilon/(1+\delta)$. Define the following objective function. Note that in the absence of the $\epsilon'$ term, it would mimic the derivative of the Nash social welfare objective from Program~\eqref{prog:3}.
\begin{equation}
\label{prog:4}
\min_{\vec{U} \in \P}  Q(\vec{U}), \text{ where } Q(\vec{U}) = \max_{\vec{U'} \in \P}  \ \sum_{i \in N} \frac{U'_i + \epsilon'}{U_i + \epsilon'}. 
\end{equation}

Clearly, $Q(\vec{U}) \ge n$ for every $\vec{U}$. Thus, the objective value in Program~\eqref{prog:4} is at least $n$. In Section~\ref{sec:fractional}, we presented an argument showing that the fractional MNW outcome is in the core. A similar argument using the first order optimality condition shows that if $\vec{U^*} \in \argmax_{\vec{U} \in \P} \sum_{i \in N} \ln (U_i+\epsilon')$, then 
$$
\sum_{i \in N} \frac{U_i + \epsilon'}{U^*_i + \epsilon'} \leq n.
$$
This implies the optimum of Program~\eqref{prog:4} is achieved at the fractional outcome maximizing the smooth Nash welfare objective $\sum_{i \in N} \ln (U_i+\epsilon')$, and this optimal value is exactly $n$.

Next, we turn to efficiently approximating Program~\eqref{prog:4}, and show that if $Q(\vec{U}) \le n(1+\delta)$, then $\vec{U}$ is a $(\delta,\epsilon)$-core outcome. 

We want to use the Ellipsoid algorithm to approximately minimize the objective function $Q(\vec{U})$ over $\vec{U} \in \P$ in polynomial time. For this, all we need is that $Q$ is a convex function, its subgradient is efficiently computable, the range of $Q$ and the diameter of $\P$ are exponentially bounded, and polytope $\P$ is efficiently separable~\cite{Bubeck}. 

First, we claim that $Q(\vec{U})$ is a convex function of $\vec{U}$. To see this, note that for any fixed $\vec{U'}$,  $U'_i/U_i$ is convex in $U_i$. Since the sum and maximum of convex functions is convex, we conclude that $Q(\vec{U})$ is also convex. 

Second, the subgradient of $Q(\vec{U})$ is efficiently computable for every $\vec{U} \in \P$. First, we find the $\vec{U'} \in \P$ that maximizes $\sum_{i\in N} \frac{U'_i + \epsilon'}{U_i + \epsilon'} $ by solving a linear program. Then, we fix $\vec{U'}$ and take the gradient of $\frac{U'_i}{U_i} $ with respect to $U_i$ to obtain a subgradient of $Q(\vec{U})$. 

Third, note that $U_i \in [0,V_{\max}]$ for each $i$. Hence, $Q(\vec{U}) \le \frac{n \cdot (V_{\max}+\epsilon')}{\epsilon'}$, which is exponentially bounded in the input size. It is easy to see that the same holds for the diameter of the polytop $\P$. 

Finally, polytope $\P$ is efficiently separable because it is a set of polynomially many linear inequalities.

Hence, we can efficiently obtain a solution $\vec{\hat{U}} \in \P$ that satisfies
$$
\max_{\vec{U'} \in \P} \sum_{i \in N} \frac{U'_i + \epsilon'}{\hat{U}_i + \epsilon'} \leq n + \delta \le n (1+ \delta).
$$

Finally, we show that $\vec{\hat{U}}$ must be a $(\delta,\epsilon)$-core outcome. Suppose for contradiction that it is not. Then, there exists a subset $S$ of agents and an outcome $\vec{U'}$ such that 
$$
(1+\delta) \cdot \hat{U}_i + \epsilon \le \frac{|S|}{n} \cdot U'_i
$$
for all $i \in S$, and at least one inequality is strict. Rearranging the terms and summing over $i \in S$, we obtain
$$
\sum_{i \in S} \frac{U'_i}{(1+\delta) \cdot \hat{U}_i + \epsilon} > |S| \cdot \frac{n}{|S|} = n.
$$

However, we also have 
$$
\sum_{i \in S} \frac{U'_i}{(1+\delta) \cdot \hat{U}_i + \epsilon} \le \sum_{i \in S} \frac{U'_i+\epsilon'}{(1+\delta) \cdot (\hat{U}_i + \epsilon')} = \frac{1}{1+\delta} \sum_{i \in S} \frac{U'_i+\epsilon'}{\hat{U}_i + \epsilon'} \le n,
$$
where the last inequality is due to approximate optimality of $\vec{\hat{U}}$. This is a contradiction. Hence, $\vec{\hat{U}}$ is a $(\delta,\epsilon)$-core outcome.

\subsection{Proof of Lemma~\ref{lem:main}}
\label{proof:main}
We first state the standard theorem for Chernoff bounds.
\begin{theorem}
\label{thm:chernoff}
Let $X_1, X_2, \ldots, X_q$ be independent  random variables in $[0,1]$, and let $X = \sum_{j=1}^q X_j$. For any $\epsilon \in (0,1)$, we have:
$$ \Pr\left[ X < (1-\epsilon) \E[X] \right] \le e^{-\frac{\epsilon^2}{2} \E[X]}$$
Equivalently, for any $\eta < \E[X]$, 
$$ \Pr\left[ X < \E[X] - \eta \right] \le e^{-\frac{\eta^2}{2 \E[X]}}$$ 
 \end{theorem}

Lemma~\ref{lem:main} follows from considering two cases. First, suppose $(1-\gamma)  A \ge  B$.
\begin{align*}
 \Pr[ X < (1-\gamma)^2 A ] & \le \Pr[ X < (1-\gamma) \E[X] ]   \le e^{- \frac{\gamma^2}{2}  \E[X] }    \\
    &  \le e^{- \frac{\gamma^2}{2} \max(\gamma B, (1-\gamma)A)} \le   e^{- \frac{\gamma^3}{2} \max(B,A)}
\end{align*}

In the other case, if $(1-\gamma)  A < B$, then $\gamma B \le \E[X] \le (1+\gamma)B$. Then
\begin{eqnarray*}
 \Pr[X < (1-\gamma) A ] & \le & \Pr[  X \le \E[X]  - \gamma B ] \\
  & \le & e^{- \gamma^2 \frac{ B ^2}{2 \E[X]}} \le  e^{- \gamma^2 \frac{ B ^2}{2(1+\gamma) B} }\\
 & \le & e^{- \frac{\gamma^3}{2} B}  \le  e^{- \frac{\gamma^3}{2} (1-\gamma)A} \le e^{- \frac{\gamma^3}{4} A}  \\
\end{eqnarray*}

\renewcommand{\O}{\mathcal{O}}

 \section{Endowments-Based Core}
 \label{sec:prop}
In this section, we show that our randomized rounding approach to packing problems extends to a slightly different definition of the core, and yields a similar approximation result. 

This alternate definition of the core only applies to packing constraints. For simplicity of presentation, we focus on the {\sc Approval Voting} setting, where $n$ agents have binary additive utilities over $m$ unit-size elements, and feasible outcomes are subsets of elements of size at most $B$. So far, we considered a notion of core in which a subset $S$ of agents can deviate and choose a feasible outcome using the entire budget $B$; however, their utility is scaled down by $|S|/n$. 

A different notion of core is based on scaling the {\em endowment}. Under this notion, when $S$ deviates, it can choose an outcome with a scaled down budget of $B \cdot |S|/n$, but then its utility is not scaled down. This notion of core has been considered in the context of participatory budgeting~\cite{Fain2016} and logically implies {\em proportional representation} of voters in multi-winner elections with approval voting. This notion builds on the seminal work of \cite{lindahlCore} on Lindahl equilibrium and its connection to the core. 


For $P \le B$, let $\O(P)$ denote the set of outcomes consisting of at most $P$ elements.  

\begin{definition}
\label{def:core2}
We say that an outcome $\bc$ is a $(\delta,\alpha)$-core outcome if for every subset $S$ of agents and every outcome $\bcp \in \O(B \cdot (1-\delta) \cdot \frac{|S|}{n})$, it is {\em not} the case that $u_i(\bcp) \geq (1+\delta) \cdot u_i(\bc) + \alpha$ for all $i \in S$ and at least one inequality is strict. We refer to a $(0,0)$-core outcome simply as a core outcome. 
\end{definition}  

As shown by~\cite{Fain2016}, it follows directly from the work of \cite{lindahlCore} that there always exists a fractional core outcome in this setting due to a fixed point argument.\footnote{Note that Theorem~\ref{thm:fractional} does not apply in this setting, since it requires scaling the utility by the size of the coalition.} However, it is not known if a fractional core or approximate core outcome can be computed in polynomial time. 

More interestingly, it is also an open question whether an integral core outcome always exists for {\sc Approval Voting}. It is easy to show that an integral core outcome does not always exist in a slightly more general setting of participatory budgeting, in which non-binary utilities and different sized elements are allowed. Consider an example with three elements $\{a,b,c\}$ of size $2$ each, a budget of $B = 3$, and three agents with {\em cyclic} preferences over the elements as follows.
\begin{center}
	\begin{tabular}{c | c c c}
		& $a$ & $b$ & $c$\\
		\hline
		$u_1$ & 1 & 0.5 & 0\\
		$u_2$ & 0 & 1 & 0.5\\
		$u_3$ & 0.5 & 0 & 1
	\end{tabular}
\end{center}

An integral outcome $\bc$ can only choose a single element. Without loss of generality, suppose $\bc = \set{a}$. Then, the set of agents $S = \set{2,3}$ and outcome $\bcp = \set{c}$ show a violation of the core. 

We now show the existence of an approximate core solution. We begin with the {\em fractional core} outcome $\vec{x}$ that can be computed using fixed point methods~\cite{lindahlCore,Fain2016}. We use dependent rounding~\cite{GandhiKS01} to round $\vec{x}$ to $\vec{X}$ so that (i) $x_j = \E[X_j]$ for each element $j$; (ii) the constraint $\sum_j X_j \le B$ is preserved; and (iii) $\{X_j\}$ are negatively correlated. Since we do not know if the fractional core outcome can be computed in polynomial time, this algorithm is not necessarily polynomial time, but it yields the following approximation result.

\begin{theorem}
\label{thm:main4}
For $\delta \in (0,1]$, there is a $(\delta,\alpha)$-core for {\sc Approval Voting}, where $ \alpha = O\left( \frac{1}{\delta^4}  \log \frac{B}{\delta}  \right)$.
\end{theorem}
\begin{proof}
We only sketch the proof of the upper bound, since it is similar to the argument in Section~\ref{sec:light}. Let $\gamma = \frac{\delta}{5}$. Let $U^*_i$ denote the utility agent $i$ receives in the fractional core outcome. We have $U^*_i \in [0,B]$. Let $\hat{U}_i$ be the random variable denoting the utility agent $i$ obtains in the rounded allocation.

Let $L =  \frac{2}{\gamma^4} \log \frac{4 B}{\gamma}$. First, if $U^*_i \le L$, then  $\hat{U}_i \ge U^*_i - L$ trivially. Otherwise, $U^*_i \ge L$, and using Lemma~\ref{lem:main}, we have:
$$ \Pr\left[ \hat{U}_i < (1- 2 \gamma) U^*_i \right]  \le  e^{- \frac{\gamma^3}{4} U^*_i}     \le e^{- \frac{\gamma^3}{4} L}  \le  \frac{1}{2 B}.$$

Let $F$ denote the subset of agents  with $\hat{U}_i < \min\left((1-2 \gamma) U^*_i, U_i^* -  L \right)$. By Markov's inequality:
$$ \Pr\left[|F| \ge \frac{n}{B}   \right] \le \frac{1}{2}.$$

Let $W = N \setminus F$. Suppose set $S$ of agents deviate. We consider two cases:

\paragraph{Case 1.} Suppose $|W \cap S| \ge (1-\gamma) |S|$. Then, consider the agents in $W \cap S$, $P = \frac{|W \cap S|}{n} B \ge (1-\gamma) \frac{|S|}{n} B \ge (1-\delta) \frac{|S|}{n} B$ and any allocation $\vec{h} \in \O(P)$. By the core condition, there exists $i \in W \cap S$ with $U^*_i \ge U_i(\vec{h})$.   Since $i \in W$, we have 
$$\hat{U}_i \ge \min\left((1-2 \gamma) U^*_i, U_i^* -  L \right).$$ 
This implies $ U_i(\vec{h}) \le \hat{U}_i (1 + 5 \gamma) + L$.
 
\paragraph{Case 2.} Otherwise, $|S \setminus W| \ge \gamma |S|$.  Then  
$$|S| \le \frac{1}{\gamma} |S \cap F| \le  \frac{1}{\gamma} |F| \le \frac{1}{\gamma}  \frac{ n}{B}.$$
Thus, if $S$ deviates, their scaled down budget is at most $1/\gamma$. Using this budget, an agent in $S$ can derive utility at most $1/\gamma < \alpha$. Since we give $\alpha$ utility for free to each agent in $S$ under our additive approximation, the approximate core condition is satisfied.
\end{proof}

The above proof generalizes to arbitrary packing constraints $A \vec{x} \le \vec{b}$. In this case, let $\Delta = \max_k \max_j \frac{b_i}{a_{kj}}$. For $P \le 1$, let $\O(P)$ denote the set of outcomes satisfying $A \vec{x} \le P \vec{b}$. 

Then, for $\delta \ge 0$ and $\alpha \ge 0$, we say that an outcome $\bc$ is a $(\delta, \alpha$)-core outcome if for any $S \subseteq N$ and outcome $\bcp \in \O(\frac{t (1 - \delta)}{n})$, it is not the case that $u_i(\bcp) \ge (1+\delta) \cdot u_i(\bc) + \alpha $ and at least one inequality is strict. Generalizing the above proof, it is easy to show the following theorem.

\begin{theorem}
For $\delta \in (0,1]$, there is a $(\delta,\alpha)$-core outcome for general packing problems, where $ \alpha = O\left( \frac{1}{\delta^4}  \log \frac{\Delta}{\delta}  \right)$.
\end{theorem}

%
%
%
%

\end{document}